\documentclass[sigconf,nonacm]{acmart}

\usepackage{amsmath}
\usepackage{amsfonts}
\usepackage{amsthm}

\usepackage{amssymb}
\usepackage{latexsym}   
\usepackage{graphicx}
\usepackage{color}
\usepackage{hyperref}
\usepackage[linesnumbered,ruled,vlined]{algorithm2e}
\usepackage{booktabs}
\usepackage{arydshln}
\usepackage{bm}
\usepackage{multirow}
\usepackage{makecell}
\usepackage{booktabs}
\usepackage{balance}
\newcommand{\argmax}{\mathop{\rm argmax}}
\newcommand{\argmin}{\mathop{\rm argmin}}

\newcommand{\degree}{\mathrm{deg}}
\newcommand{\rd}{\, \mathrm{d}}

\newcommand{\hide}[1]{}

\newtheorem{lemma}{Lemma}
\newtheorem{theorem}{Theorem}
\newtheorem{corollary}{Corollary}
\newtheorem{problem}{Problem}

\newenvironment{newprocedure}[1][ht]{%
   \begin{algorithm}[#1]%
  }{\end{algorithm}}

\AtBeginDocument{%
  }

\copyrightyear{2023}
\acmYear{2023}
\setcopyright{acmlicensed}\acmConference[KDD '23]{Proceedings of the 29th ACM SIGKDD Conference on Knowledge Discovery and Data Mining}{August 6--10, 2023}{Long Beach, CA, USA}
\acmBooktitle{Proceedings of the 29th ACM SIGKDD Conference on Knowledge Discovery and Data Mining (KDD '23), August 6--10, 2023, Long Beach, CA, USA}
\acmPrice{15.00}
\acmDOI{10.1145/3580305.3599306}
\acmISBN{979-8-4007-0103-0/23/08}

\begin{document}

\title{Densest Diverse Subgraphs:\\ How to Plan a Successful Cocktail Party with Diversity}

\author{Atsushi Miyauchi}
\email{atsushi.miyauchi@centai.eu}
\orcid{0000-0002-6033-6433}
\affiliation{%
  \institution{CENTAI Institute}
  \city{Turin}
  \country{Italy}
}

\author{Tianyi Chen}
\email{ctony@bu.edu}
\affiliation{%
  \institution{Boston University}
  \city{Boston}
  \state{MA}
  \country{USA}
}

\author{Konstantinos Sotiropoulos}
\email{ksotirop@bu.edu}
\affiliation{%
  \institution{Boston University}
  \city{Boston}
  \state{MA}
  \country{USA}
}

\author{Charalampos E. Tsourakakis}
\email{ctsourak@bu.edu}
\affiliation{%
  \institution{Boston University}
  \city{Boston}
  \state{MA}
  \country{USA}
}

\renewcommand{\shortauthors}{Miyauchi et al.}

\begin{abstract}
Dense subgraph discovery methods are routinely used in a variety of applications including the identification of a team of skilled individuals for collaboration from a social network. 
However, when the network's node set is associated with a sensitive attribute such as race, gender, religion, or political opinion, the lack of diversity can lead to lawsuits. 

In this work, we focus on the problem of finding a densest diverse subgraph in a graph whose nodes have different attribute values/types that we refer to as colors. We propose two novel formulations motivated by different realistic scenarios. Our first formulation, called the \emph{densest diverse subgraph problem} (DDSP), guarantees that no color represents more than some fraction of the nodes in the output subgraph, which generalizes the state-of-the-art due to Anagnostopoulos et al. (CIKM~2020). By varying the fraction we can range the diversity constraint and interpolate from a diverse dense subgraph where all colors have to be equally represented to an unconstrained dense subgraph.  We design a scalable $\Omega(1/\sqrt{n})$-approximation algorithm, where $n$ is the number of nodes. Our second formulation is motivated by the setting where any specified color should not be overlooked. We propose the \emph{densest at-least-$\vec{k}$-subgraph problem} (Dal$\vec{k}$S), a novel generalization of the classic Dal$k$S, where instead of a single value $k$, we have a vector $\bm{k}$ of cardinality demands with one coordinate per color class. We design a $1/3$-approximation algorithm using linear programming together with an acceleration technique. Computational experiments using synthetic and real-world datasets demonstrate that our proposed algorithms are effective in extracting dense diverse clusters.

\end{abstract}





\maketitle

\section{Introduction}
\label{sec:intro} 
Dense subgraph discovery (DSD) is a fundamental graph-mining primitive, routinely used to mine social, financial, and biological networks among others~\cite{Gionis_Tsourakakis_15}. Applications include team formation~\cite{gajewar2012multi,rangapuram2013towards}, detecting correlated genes~\cite{tsourakakis2013denser}, community and spam link farm detection in the Web graph~\cite{Gibson+05,Dourisboure+07}, finding experts  in crowdsourcing systems~\cite{Kawase+19}, spotting money laundering in financial networks~\cite{Li2020FlowScopeSM,starnini2021smurf,chen2022antibenford}, assessing the statistical significance of motifs~\cite{chen2023algorithmic}, and modeling real-world networks~\cite{chanpuriya2021power}. See the tutorial~\cite{Gionis_Tsourakakis_15} for an extensive list of related applications. 

Among various DSD formulations, the densest subgraph problem (DSP) stands out for various reasons~\cite{Lanciano+23}. For a given undirected graph $G=(V,E)$ with $n=|V|$ nodes and $m=|E|$ edges, DSP is solvable in polynomial time using maximum flow~\cite{Goldberg84,Picard82} or linear programming (LP)~\cite{Charikar00} and can be approximated within a factor of 2 using a greedy algorithm~\cite{Kortsarz+94,Charikar00}, 
and more recently solved near-optimally by an iterative greedy algorithm over the input~\cite{chekuri2022densest,boob2020flowless}. 
Furthermore, DSP is solvable on massive graphs using distributed and streaming implementations~\cite{Bahmani+12} 
and admits useful variants, e.g.,~\cite{Kawase_Miyauchi_18,Kuroki+20,Tsourakakis_15,Veldt+21}. 

In numerous real-world settings, we have extra information about the nodes. For example, in a graph database we may know each individual's gender and race. On the Twitter follow network, there exist methods to infer from tweets whether a node is positive, neutral, or negative towards a controversial topic~\cite{de2016learning}.  On brain networks neurons can play different functional roles~\cite{adami2011information,matejek2022edge}. 
Consider also the problem of organizing a cocktail party~\cite{sozio2010community} with diversity constraints as follows:

\begin{quote}
A number of computer scientists aim to organize a cocktail party to celebrate Turing's legacy. They believe that the success of the event will be higher if they invite computer scientists who have collaborated in the past  but also who span different research areas. Whom should they invite?
\end{quote} 

We will refer to the set of different attribute values/types of nodes as \emph{colors}. The aforementioned   settings motivate the problem of finding a \emph{densest diverse subgraph}, namely a subset of nodes that induce many edges, but also are diverse in terms of colors. We will be referring to the concept of diversity as \emph{fairness} when the attribute concerns sensitive information such as gender, race, or religion. Applying a DSD method does not guarantee that the extracted densest subgraph will be diverse. Actually on a variety of real data, we observe that this is the typical case, i.e., the densest subgraph often exhibits strong homophily. Suppose the output of such a method is used to select a group of individuals in a social network. In that case, it will not be representative of the different races/religions/opinions that may co-exist in the network. This can be especially harmful in the context of selecting teams using dense subgraphs~\cite{rangapuram2013towards}, recommending material to social media users that is not balanced in terms of opinions and hence increasing polarization~\cite{musco2018minimizing} or even leading to lawsuits~\cite{minisi}. In such cases, it becomes of paramount importance to have algorithms that can extract a cluster with diversity guarantees.   

With the exception of a recent paper by Anagnostopoulos et al.~\cite{anagnostopoulos2020spectral}, very little attention has been given to dense diverse subgraph extraction, despite the extensive research on the DSP and its applications~\cite{Lanciano+23}. Although Anagnostopoulos et al.~\cite{anagnostopoulos2020spectral} have made progress, there are still many unanswered questions. For example, the methods they propose exclusively focus on scenarios involving two colors and strive to achieve the complete fairness in the output. Furthermore, their spectral approach offers theoretical guarantees, but these guarantees are contingent upon restrictive conditions for the spectrum of the input graph that are computationally burdensome to verify. Additionally, they offer heuristics for scenarios involving more than two colors, but without any assurances regarding their quality.

\subsection{Our Contributions}
We introduce two novel formulations for finding a densest diverse subgraph. The first one is called the \emph{densest diverse subgraph problem} (DDSP), and the second one is called the \emph{densest at-least-$\vec{k}$-subgraph problem} (Dal$\vec{k}$S).  Informally, the first problem aims to offer diversity guarantees that concern the relative sizes of the color classes, while the second guarantees in terms of absolute counts. 

Let $C$ be the set of colors. Our first formulation guarantees the diversity of the output in the sense that  it is not dominated by any single color. To this end, the formulation introduces a parameter $\alpha \in [1/|C|,1]$ that determines the maximum portion of any color in the output solution. It should be noted that our formulation is a substantial generalization of the \emph{fair densest subgraph problem} introduced by Anagnostopoulos et al.~\cite{anagnostopoulos2020spectral}, enabling us to deal with non-binary attributes and to freely adjust the degree of diversity.  
Interestingly, from a theory perspective, our formulation contains two important variants of DSP, the densest at-least-$k$-subgraph problem (Dal$k$S) and the densest at-most-$k$-subgraph problem (Dam$k$S), as special cases, cf.  Sections~\ref{sec:related} and~\ref{subsec:dds} for details. 

On the other hand, our second formulation guarantees the diversity of the output in the sense that it does not overlook any specified color. In particular, consider a graph where $|C|$ is a very small constant and the minority colors appear only in a handful of nodes. Instead of imposing {\em relative} constraints on the sizes through ratios, we impose {\em absolute} constraints on the cardinalities of the nodes from each color.  Specifically, the formulation requires the output to contain at least a given number of representatives from each color. The formulation is a novel generalization of Dal$k$S, where instead of just demanding $k$ nodes in the output, we have a vector $\bm{k}$ of demands from each color, i.e., they are lower bounds on how many nodes we have to include from each possible color.   

As both formulations are NP-hard, we design polynomial-time approximation algorithms: 
For the first problem, we provide an approximation algorithm for the case where $V$ is already diverse (i.e., $V$ is a feasible solution for the problem). 
Our algorithm has an approximation ratio of $\gamma \cdot \max\left\{ \frac{1}{\lceil 1/\alpha \rceil},\, \frac{1}{\alpha n}\right\}$, where $\gamma$ is the best approximation ratio known for Dal$k$S (currently equal to $1/2$~\cite{Khuller_Saha_09}). 
By simple calculation, we observe that the above approximation ratio leads to an approximation ratio of $\Omega(1/\sqrt{n})$, irrespective of any parameter other than the number of nodes. 
Moreover, we can also see that the approximation ratio is lower bounded by $1/|C|$, meaning that the algorithm attains a constant-factor approximation for the case of $|C|=O(1)$ and is a generalization of the $1/2$-approximation algorithm for the fair densest subgraph problem by Anagnostopoulos et al.~\cite{anagnostopoulos2020spectral}. 
Our algorithm is based on an approximation algorithm for Dal$k$S with a carefully selected value of $k$, together with a simple postprocessing. The primary factor determining the time complexity of our algorithm is the time complexity of the approximation algorithm used for Dal$k$S.

For the second problem, we devise a $1/3$-approximation algorithm, 
which runs in polynomial time for constant $|C|$.  In the design and analysis of our algorithm, we generalize the existing $1/2$-approximation algorithm for Dal$k$S and its approximation ratio analysis. 
As shown later, we can get a $1/4$-approximate solution directly using the $1/2$-approximation algorithm for Dal$k$S. 
Our effort improves the approximation ratio from $1/4$ to $1/3$, by sacrificing some degree of scalability.   We also present an acceleration technique for the proposed algorithm with the aid of the well-known greedy peeling algorithm~\cite{Charikar00}. The running time of our original algorithm is $O((n/|C|)^{|C|}T_\text{LP})$, where $T_\text{LP}$ is the time required to solve an LP used in the algorithm, while the accelerated version runs in $O(|C|(n/|C|)^{|C|-1}T_\text{LP})$. In the case where $|C|$ is a  small constant, the reduction of the running time due to the acceleration is drastic. 

We evaluate our algorithms on real-life attributed datasets, including social networks with gender and profession attributes. We compare against Anagnostopoulos et al.~\cite{anagnostopoulos2020spectral}, but we also develop a novel baseline that uses node embeddings~\cite{perozzi2014deepwalk,qiu2018network,chanpuriya2020node} combined with advances in scalable fair clustering of points~\cite{backurs2019scalable}. The algorithms we propose have the capability to extract dense and diverse subgraphs. We demonstrate that real-world networks contain dense subgraphs that exhibit significant homophily, emphasizing the importance of employing our tools in scenarios where diversity or fairness is essential.

\section{Related work}
\label{sec:related} 
\noindent \textbf{DSP and its variations.} Given an undirected graph $G=(V,E)$, we define for any non-empty subset of nodes $S\subseteq V$  the (degree) density $d(S)=|E(S)|/|S|$, where $E(S)=\{\{u,v\}\in E\mid u,v\in S\}$. The DSP aims to maximize the degree density over all possible subsets. Notice that the degree density is just half of the average degree of an induced subgraph. The DSP is polynomial-time solvable using  $O(\log n)$ max-flow computations~\cite{Goldberg84,Picard82}, $O(1)$ number of flows using parametric maximum flow~\cite{gallo1989fast}, or LP~\cite{Charikar00}. There also exists a linear-time $1/2$-approximation algorithm that greedily removes a node of the smallest degree, and reports the maximum degree density seen among these subsets~\cite{Kortsarz_Nutov_05,Charikar00}. 
This kind of algorithm is often called the greedy peeling algorithm. 
Recently, Boob et al.~\cite{boob2020flowless} proposed {\sc Greedy++}, an iterative peeling algorithm that generalizes the above and converges to a near-optimal solution extremely fast without the use of flows. Very recently, Chekuri et al. \cite{chekuri2022densest} analyzed the performance of {\sc Greedy++}, and showed that its algorithmic idea could be generalized to general fractional supermodular maximization.  

DSP has a lot of problem variants~\cite{Lanciano+23}. 
Unlike the original DSP, its size-restricted variants are known to be NP-hard. Indeed, the densest $k$-subgraph problem (D$k$S) that asks for the densest subgraph with exactly $k$ nodes, is not only NP-hard, but also hard to approximate, with the best-known approximation ratio being
$\Omega(1/n^{1/4+\epsilon})$ for any $\epsilon>0$~\cite{Bhaskara+10}. This approximability result is far off from the best-known hardness result that assumes the Exponential Time Hypothesis (ETH). If ETH holds, then D$k$S cannot be approximated within a ratio of $n^{1/(\log\log n)^c}$ for some $c>0$.  
Another size-restricted variant, Dam$k$S, aims to maximize the degree density over all subsets of nodes $S$ such that $|S|\leq k$~\cite{Andersen_Chellapilla_09}. 
It is known that $\alpha$-approximation to Dam$k$S leads to $\alpha/4$-approximation to D$k$S~\cite{Khuller_Saha_09}. 

Closest to this work lies Dal$k$S, which imposes the cardinality constraint $|S|\geq k$~\cite{Andersen_Chellapilla_09}.  The problem is also known to be NP-hard~\cite{Khuller_Saha_09}. Andersen and Chellapilla~\cite{Andersen_Chellapilla_09} designed a $1/3$-approximation algorithm that runs in linear time, using greedy peeling. 
Khuller and Saha~\cite{Khuller_Saha_09} designed two different approximation algorithms, that both achieve $1/2$-approximation using either a small number of flows, or by solving an LP~\cite{Khuller_Saha_09}. 
 
\smallskip
\noindent \textbf{Fair densest subgraph problem.}  Despite the large amount of research on DSD, the problem of finding a densest diverse subgraph has not received attention with the single exception of Anagnostopoulos et al.~\cite{anagnostopoulos2020spectral} who introduced the fair densest subgraph problem for two colors. 
Assuming the graph is fair to begin with, i.e., the two colors are equally represented in $V$, they demand equal representation of each category in the output. 
For this case, they proposed a greedy $1/2$-approximation algorithm and a spectral approach. The spectral approach comes with guarantees only in limited cases (e.g., all degrees being almost equal), which are not typical on real data that tend to have a skewed degree distribution. This algorithm can be extended to the case of $|C|>2$, but without quality guarantees. 
Finally, the authors studied the hardness of the problem and showed that their formulation is at least as hard as Dam$k$S: Any $\alpha$-approximation to their formulation leads to $\alpha$-approximation to Dam$k$S. 

\smallskip
\noindent \textbf{Fairness and algorithms.} While DSD with diversity is not yet well studied, fair clustering of clouds of points has received much attention from the data mining community. Chierichetti et al.~\cite{chierichetti2017fair} initiated the problem of finding balanced clusters in a cloud of points, namely clusters where two groups are equally represented. They proposed a method called fairlet decomposition, that decomposes a dataset into minimal sets that satisfy fair representation, called fairlets. Afterwards, typical machine learning algorithms, like $k$-median, can be used to obtain fair clusters. Backurs et al.~ \cite{backurs2019scalable}  provided a scalable algorithm for the fairlet decomposition. Later work has extended the problem of finding fair clusters to the case of correlation clustering~\cite{ahmadian2020fair_cor} and hierarchical clustering ~\cite{ahmadian2020fair_hier}.

\section{Problem statements} 
\label{sec:statements}
In this section, we formally introduce our optimization problems. 
Let $G=(V,E)$ be an undirected graph with $n=|V|$ nodes and $m=|E|$ edges. Let $C$ be a set of colors. Without loss of generality $|C|\leq n$. Let $\ell \colon V\rightarrow C$ be the coloring function that assigns a color to each node.  Given the above as input, we aim to find a densest diverse subgraph. We mathematically formalize the notion of diversity in two ways found in Sections~\ref{subsec:dds} and~\ref{subsec:dalks}, respectively.

\subsection{Densest Diverse Subgraph Problem (DDSP)}
\label{subsec:dds}

Our first notion aims to ensure that no single color dominates the rest. Specifically, for $S\subseteq V$, we denote by $c_\text{max}(S)$ the maximum number of nodes in $S$  that receive the same color, i.e., $c_\mathrm{max}(S)=\max_{c\in C} |\{v\in S\mid \ell(v)=c\}|$. 
We also denote by $\alpha(S)$ the maximum fraction of monochromatic nodes in $S$, i.e., $\alpha(S)=c_\mathrm{max}(S)/|S|$. Our problem can be formulated as follows:  

\begin{problem}[DDSP]
\label{prob:fair}
Given an undirected graph $G=(V,E)$, $\ell \colon V\rightarrow C$, and $\alpha \in [1/|C|,1]$,  find a subset of nodes $S\subseteq V$ that maximizes the degree density $d(S)$ subject to the constraint $\alpha(S)\leq \alpha$. 
\end{problem}

This problem is a major generalization of the fair densest subgraph problem introduced by Anagnostopoulos et al.~\cite{anagnostopoulos2020spectral}, which is obtained for the special values $|C|=2$ and $\alpha=1/2$. 
As the fair densest subgraph problem is NP-hard~\cite{anagnostopoulos2020spectral}, Problem~\ref{prob:fair} is also NP-hard. Clearly when $\alpha=1$ we are oblivious to diversity and obtain (polynomial-time solvable) DSP.  
More interestingly, Problem~\ref{prob:fair} contains two totally different optimization problems, Dal$k$S and Dam$k$S, as special cases: 
\begin{proposition}\label{prop:reduction} 
There exist polynomial-time reductions from Dal$k$S and Dam$k$S to Problem~\ref{prob:fair}.
\end{proposition} 

\begin{proof} 
The reductions are obtained by appropriately setting the number of colors and the parameter $\alpha$.  For Dal$k$S, it suffices to construct the instance of Problem~\ref{prob:fair} by coloring the nodes with $n$ distinct colors and setting $\alpha$ to be $1/k$. For Dam$k$S, it suffices to construct the instance of Problem~\ref{prob:fair} by coloring the nodes with the single color and adding $k$ isolated nodes (i.e., dummy nodes) with another color, and setting $\alpha=1/2$. 
\end{proof}

\subsection{Densest At-Least-$\vec{k}$-Subgraph (Dal$\vec{k}$S)}
\label{subsec:dalks}

Our second formulation diversifies the output by ensuring that it does not overlook any specified color. 
To this end, the formulation requires the output to contain at least a given number of representatives from each color. 
For $S\subseteq V$ and $c\in C$, let $S_c=\{v\in S\mid \ell(v)=c\}$. 
Our problem formulation is as follows: 
\begin{problem}[{\rm Dal$\vec{k}$S}]\label{prob:dalkks}
Given an undirected graph $G=(V,E)$, $\ell \colon V\rightarrow C$, and $\bm{k}=(k_c)_{c\in C}\in \mathbb{Z}_{\geq 0}^{|C|}$, find a subset of nodes $S\subseteq V$ that maximizes the degree density $d(S)$ subject to $|S_c|\geq k_c$ for any $c\in C$. 
\end{problem}

Obviously the above problem is a generalization of Dal$k$S. 
Thus, the problem is NP-hard. 
If $\bm{k}=\bm{0}$, the problem is reduced to the original DSP; 
therefore, throughout the paper, we assume that $k_c\geq 1$ for some $c\in C$. 
As mentioned in the introduction, we can easily get a $1/4$-approximation algorithm for the problem: 

\begin{proposition}\label{prop:baseline}
For Problem~\ref{prob:dalkks}, there exists a polynomial-time $1/4$-approximation algorithm. 
\end{proposition}

\begin{proof}
Let $G = (V, E)$, $\ell \colon V \rightarrow C$, $\bm{k}=(k_c)_{c \in C}\in \mathbb{Z}_{\geq 0}^{|C|}$ be an instance of Problem 2, and $\text{OPT}$ the optimal value of the instance. To get a feasible solution for Problem 2, we have to take at least $k_c$ nodes for every color $c \in C$; therefore, Dal$k$S with $k = \|\bm{k}\|_1=\sum_{c\in C}k_c$ on $G$ is a relaxation of Problem 2 on $G$. Let $S \subseteq V$ be an $\alpha$-approximate solution for Dal$k$S with $k = \|\bm{k}\|_1$. As Dal$k$S with $k = \|\bm{k}\|_1$ is a relaxation of Problem 2, we have $d(S) \geq \alpha \cdot \text{OPT}$. Note that $S$ is not necessarily feasible for Problem 2, but we can make it feasible by adding at most $k_c$ nodes for every color $c \in C$, resulting in adding at most $\|\bm{k}\|_1 \leq |S|$ nodes in total. Letting $S’ \subseteq V$ be the resulting subset, we have $d(S’) = \frac{|E(S’)|}{|S’|} \geq \frac{|E(S)|}{|S’|} \geq \frac{1}{2}\cdot d(S) \geq \frac{\alpha}{2}\cdot \text{OPT}$, meaning that $S’$ is an $\alpha/2$-approximate solution for Problem 2.
As mentioned above, there is a polynomial-time $1/2$-approximation algorithm for Dal$k$S~\cite{Khuller_Saha_09}. 
Therefore, we can set $\alpha=1/2$ and have a polynomial-time $1/4$-approximation algorithm for Problem~\ref{prob:dalkks}.
\end{proof}

\section{Algorithm for Problem~\ref{prob:fair}}
\label{sec:fair}

In this section, we design a polynomial-time $\Omega(1/\sqrt{n})$-approximation algorithm for Problem~\ref{prob:fair}. 
A subset of nodes $S\subseteq V$ is said to be diverse if $\alpha(S)\leq \alpha$ holds. 
In what follows, we  assume that $V$ is diverse, i.e., $\alpha(V) \leq \alpha$.

\subsection{The Proposed Algorithm}\label{subsec:algo_approx}
Our algorithm first computes a constant-factor approximate solution (say $\gamma$-approximate solution) to Dal$k$S 
on $G$ with $k=\lceil 1/\alpha \rceil$. 
For example, we can use a $1/2$-approximation algorithm using LP~\cite{Khuller_Saha_09} or 
a $1/3$-approximation algorithm using greedy peeling~\cite{Andersen_Chellapilla_09}. 
Then, the algorithm makes the solution feasible by adding an arbitrary node with a color of the lowest participation iteratively (Procedure~\ref{alg:diversify}). 
For reference, the entire procedure of our algorithm is summarized in Algorithm~\ref{alg:approx}. 

The time complexity of Algorithm~\ref{alg:approx} is dominated by the algorithm we use for Dal$k$S. 
Even if we consider the objective value in Procedure~\ref{alg:diversify}, it still depends on the choice of the approximation algorithm for Dal$k$S. 
If we employ a $1/2$-approximation algorithm using LP, the time complexity of Algorithm~\ref{alg:approx} is dominated by that required for solving the LP. 
If we use a $1/3$-approximation algorithm using greedy peeling, 
Algorithm~\ref{alg:approx} can be implemented to run in $O(m+n\log n)$ time, by handling the nodes outside $S$ using a Fibonacci heap for each color with key values being degrees to $S$. 

\begin{newprocedure}[t]
\caption{$\mathsf{Diversify}(S)$}\label{alg:diversify}
\SetKwInOut{Input}{Input}
\SetKwInOut{Output}{Output}
\While{$\alpha(S)>\alpha$}{
Find $v_\mathrm{min}\in V\setminus S$ that satisfies $\ell(v_\mathrm{min})\in \argmin_{c\in C}|S_c|$\;
\tcc{In practice consider also the objective value.}
$S\leftarrow S\cup \{v_\mathrm{min}\}$\;
}
\Return $S$\;
\end{newprocedure}

\begin{algorithm}[t]
\caption{$\Omega(1/\sqrt{n})$-approximation algorithm}\label{alg:approx}
\SetKwInOut{Input}{Input}
\SetKwInOut{Output}{Output}
\Input{\ $G=(V,E)$, $\ell \colon V\rightarrow C$, $\alpha \in [1/|C|,1]$}
\Output{\ $S\subseteq V$}
$S\leftarrow$  $\gamma$-approximate solution to Dal$k$S on $G$ with $k=\lceil 1/\alpha \rceil$\;
\tcc{See \cite{Khuller_Saha_09} and \cite{Andersen_Chellapilla_09} for the algorithms achieving $\gamma=1/2$ and $\gamma=1/3$, respectively.}
\Return $\mathsf{Diversify}(S)$\;
\end{algorithm}

\subsection{Analysis}
We analyze the approximation ratio of Algorithm~\ref{alg:approx}. We have the following key lemma: 

\begin{lemma}\label{lem:size_ub}
Assume that $\alpha(V)\leq \alpha$ holds. Then, for any $S\subseteq V$ with $|S|\geq \lceil 1/\alpha \rceil$, it holds that $|\mathsf{Diversify}(S)|\leq \min\{\lceil 1/\alpha \rceil, \alpha n\}\cdot |S|$. 
\end{lemma}
\begin{proof}
We first prove that $c_\mathrm{max}(S)=c_\mathrm{max}(\mathsf{Diversify}(S))$ holds. 
As $S\subseteq \mathsf{Diversify}(S)$, we have $c_\mathrm{max}(S)\leq c_\mathrm{max}(\mathsf{Diversify}(S))$. 
Therefore, it suffices to show $c_\mathrm{max}(S)\geq c_\mathrm{max}(\mathsf{Diversify}(S))$. 
For $S \subseteq V$ and $c \in C$, we denote by $f(S,c)$ the fraction of the nodes in $S$ that receive the color $c$, i.e., $f(S,c) = |S_c|/|S|$. Let us focus on an arbitrary iteration of Procedure~\ref{alg:diversify} and let $S’ \subseteq V$ be the subset kept at the beginning of the iteration. Then there exists $c \in C$ that satisfies the condition $f(S’,c) < \alpha$. Suppose, for contradiction, that there exist no such colors. Then, for any color $c \in C$, we have $f(S’,c) \geq \alpha$. Moreover, as $S’$ is not yet feasible, we see that there exists $c’ \in C$ that satisfies $f(S’,c’) > \alpha$. Therefore, we have 
\begin{align*}
1 = \sum_{c\in C} f(S’,c) &= f(S’,c’) + \sum_{c \in C \setminus \{c’\}} f(S’,c)\\ 
  &> \alpha + (|C| - 1)\alpha = |C|\alpha \geq 1,
\end{align*}
a contradiction. 
From the above, recalling the greedy rule of Procedure~\ref{alg:diversify}, we see that Procedure~\ref{alg:diversify} only adds the nodes with the colors $c \in C$ that satisfy $f(S’,c) < \alpha$. To increase $c_\text{max}$ in this iteration, Procedure~\ref{alg:diversify} needs to add a node with a color $c \in C$ that satisfies $f(S’,c) > \alpha$, but it does not happen. As we fixed an iteration arbitrarily, $c_\text{max}$ does not increase throughout Procedure~\ref{alg:diversify}. Therefore, we have $c_\text{max}(S) \geq c_\text{max}(\textsf{Diversify}(S))$. 

Assume that in some iteration, which produces $S''\subseteq V$, of Procedure~\ref{alg:diversify}, $|S''|= \lceil |S|/\alpha \rceil$ holds. 
Then we have 
\begin{align*}
\alpha(S'')=\frac{c_\mathrm{max}(S'')}{|S''|}
= \frac{c_\mathrm{max}(S)}{\lceil |S|/\alpha \rceil}
\leq \alpha \cdot \frac{c_\mathrm{max}(S)}{|S|}\leq \alpha, 
\end{align*}
where the second equality follows from $c_\mathrm{max}(S)= c_\mathrm{max}(\mathsf{Diversify}(S))$. 
This means that the algorithm terminates at or before this iteration. 
Therefore, we have $|\mathsf{Diversify}(S)|\leq \lceil |S|/\alpha \rceil$. 

Similarly, assume that in some iteration, which produces $S''\subseteq V$, of Procedure~\ref{alg:diversify}, $|S''|= c_\mathrm{max}(S)|S|$ holds. Then we have 
\begin{align*}
\alpha(S'')=\frac{c_\mathrm{max}(S'')}{|S''|}
=\frac{c_\mathrm{max}(S)}{c_\mathrm{max}(S)|S|}
=\frac{1}{|S|}
\leq \frac{1}{\lceil 1/\alpha \rceil}
\leq \alpha, 
\end{align*}
where the first inequality follows from the assumption of the lemma. 
This means that the algorithm terminates at or before this iteration. 
Therefore, we have $|\mathsf{Diversify}(S)|\leq c_\mathrm{max}(S)|S|$. 

From the above, we see that 
\begin{align*}
|\mathsf{Diversify}(S)|&\leq \min\{\lceil|S|/\alpha \rceil, c_\mathrm{max}(S)|S|\}\\
&\leq \min\{\lceil 1/\alpha \rceil, c_\mathrm{max}(V)\}\cdot |S|\\
&\leq \min\{\lceil 1/\alpha \rceil, \alpha n\}\cdot |S|, 
\end{align*}
which completes the proof. 
\end{proof}

Using the above lemma, we can prove the following: 

\begin{theorem} \label{thm:alg2_appx}
Algorithm~\ref{alg:approx} is a $\left(\gamma \cdot \max\left\{\frac{1}{\lceil 1/\alpha \rceil},\, \frac{1}{\alpha n}\right\}\right)$-approximation algorithm for Problem~\ref{prob:fair} 
when $\alpha(V)\leq \alpha$ holds. Here $\gamma$ is the approximation ratio of the algorithm for Dal$k$S used in Algorithm~\ref{alg:approx}. 
\end{theorem}

\begin{proof}
Let $\mathsf{OPT}$ be the optimal value of Problem~\ref{prob:fair}. 
Let $S\subseteq V$ be a constant-factor approximate solution (say $\gamma$-approximate solution) for Dal$k$S on $G$ with $k=\lceil 1/\alpha \rceil$. 
Note that Dal$k$S with $k=\lceil 1/\alpha \rceil$ is a relaxation of Problem~\ref{prob:fair}, because even if we pick the nodes, all of which have different colors, we need at least $k=\lceil 1/\alpha \rceil$ nodes to satisfy the diversity constraint. Thus, we have $d(S)=\frac{|E(S)|}{|S|}\geq \gamma \cdot \mathsf{OPT}$. The output of the algorithm is $\mathsf{Diversify}(S)$ whose objective value can be evaluated as follows: 
\begin{align*}
d(\mathsf{Diversify}(S))
&=\frac{|E(\mathsf{Diversify}(S))|}{|\mathsf{Diversify}(S)|}\\
&\geq \frac{|E(S)|}{\min\{\lceil 1/\alpha \rceil, \alpha n\}\cdot |S|}\\
&\geq \gamma \cdot \max\left\{\frac{1}{\lceil 1/\alpha \rceil},\, \frac{1}{\alpha n}\right\} \cdot \mathsf{OPT}, 
\end{align*}
where the first inequality follows from Lemma~\ref{lem:size_ub}. 
\end{proof}

It should be emphasized that the approximation ratio of Algorithm~\ref{alg:approx} is lower bounded by $1/|C|$, meaning that the algorithm is a constant-factor approximation algorithm for the case of $|C|=O(1)$. 
In fact, we see that $\frac{1}{\lceil 1/\alpha \rceil}\geq \frac{1}{\lceil|C|\rceil}=\frac{1}{|C|}$. 
Therefore, Algorithm~\ref{alg:approx} can be seen as a generalization of the $1/2$-approximation algorithm for the fair densest subgraph problem by Anagnostopoulos et al.~\cite{anagnostopoulos2020spectral}. 

Moreover, the above analysis leads to an approximation ratio that is independent of any parameter other than $n$: 

\begin{corollary}
Algorithm~\ref{alg:approx} is an $\Omega(1/\sqrt{n})$-approximation algorithm for Problem~\ref{prob:fair} when $\alpha(V)\leq \alpha$ holds. 
\end{corollary}
\begin{proof}
We can lower bound the approximation ratio given in Theorem~\ref{thm:alg2_appx} as follows:  
\begin{align*}
\gamma \cdot \max\left\{\frac{1}{\lceil 1/\alpha \rceil},\, \frac{1}{\alpha n}\right\}
\geq \gamma \cdot \sqrt{\frac{1}{\lceil 1/\alpha \rceil}\cdot \frac{1}{\alpha n}}=\Omega\left(\frac{1}{\sqrt{n}}\right).  
\end{align*}
\end{proof}

\section{Algorithm for Problem~\ref{prob:dalkks}}
\label{sec:alg_prob2}

In this section, we design a polynomial-time $1/3$-approximation algorithm for Problem~\ref{prob:dalkks} with $|C|=O(1)$.  
Using a sophisticated LP, we can improve the  approximation ratio of $1/4$ given in Proposition~\ref{prop:baseline}. Recall that $\bm{k}=(k_c)_{c\in C}\in \mathbb{Z}_{\geq 0}^{|C|}$ is the input lower-bound vector.

\subsection{The Proposed Algorithm}
Let $\bm{p}=(p_c)_{c\in C}\in \mathbb{Z}_{\geq 0}^{|C|}$ be a vector that satisfies $\bm{p} \geq \bm{k}$.  
We consider the following LP: 
\begin{alignat*}{4}
&\mathrm{LP}(\bm{p})\colon \ \ &\mathrm{maximize}&\ \ &   &\sum_{e\in E}x_e\\
&                       &\mathrm{subject\ to}&\ \ &   &x_e\leq y_u, \ x_e\leq y_v &\quad &\forall e=\{u,v\}\in E,\\
&&&&&\sum_{v\in V_c}y_v = \frac{p_c}{\|\bm{p}\|_1} &&\forall c\in C,\\
&&&&& y_v \leq \frac{1}{\|\bm{p}\|_1}&  &\forall v\in V,\\ 
&&&&&x_e, y_v \geq 0 &  &\forall e\in E, \, \forall v\in V.
\end{alignat*}
Note that this is a major generalization of the LP used in the $1/2$-approximation algorithm for Dal$k$S~\cite{Khuller_Saha_09}. 
The LP for Dal$k$S is parameterized by a single value, while our LP is parameterized by multiple values (i.e., a vector $\bm{p}$). 
This modification is essential to address the generalization caused by Dal$\vec{k}$S. 

Let $(\bm{x}^*,\bm{y}^*)$ be an optimal solution to $\mathrm{LP}(\bm{p})$ for $\bm{p}$. 
For $(\bm{x}^*,\bm{y}^*)$ and $r\geq 0$, we define $S(r)=\{v\in V\mid y^*_v\geq r\}$. 
Now we are ready to present our algorithm. 
For each $\bm{p}\in \mathbb{Z}_{\geq 0}^{|C|}$ such that $\bm{k}\leq \bm{p}\leq (|V_c|)_{c\in C}$, 
the algorithm conducts the following procedure: 
It first solves $\mathrm{LP}(\bm{p})$ to obtain an optimal solution $(\bm{x}^*,\bm{y}^*)$. 
Then using the solution, the algorithm enumerates all possible $S(r)$'s by setting $r\in \{y^*_v\mid v\in V\}\cup \{0\}$ each, makes them feasible by adding nodes with appropriate colors using Procedure~\ref{alg:dalkks_subroutine}, and takes the best subset among them, as a candidate for $\bm{p}$. 
After the above iterations, the algorithm finally outputs the best subset among all candidates for $\bm{p}$ with $\bm{k}\leq \bm{p}\leq (|V_c|)_{c\in C}$. 
The entire procedure is summarized in Algorithm~\ref{alg:dalkks}. 

\begin{newprocedure}[t]
\caption{$\textsf{Make\_it\_feasible}(S)$}\label{alg:dalkks_subroutine}
\For{each $c\in C$}{
  \If{$|S_c|<k_c$}{
    Take arbitrary $v\in V_c\setminus S_c$\;
    \tcc{In practice consider also the objective value.}
    $S\leftarrow S\cup \{v\}$\;
    }
}
\Return $S$\;
\end{newprocedure}

\begin{algorithm}[t]
\caption{$1/3$-approximation algorithm for Problem~\ref{prob:dalkks}}\label{alg:dalkks}
\SetKwInOut{Input}{Input}
\SetKwInOut{Output}{Output}
\Input{\ $G=(V,E)$, $\ell \colon V\rightarrow C$, $\bm{k}=(k_c)_{c\in C}\in \mathbb{Z}_{\geq 0}^{|C|}$}
\Output{\ $S\subseteq V$}
\For{each $\bm{p}$ such that $\bm{k}\leq \bm{p}\leq (|V_c|)_{c\in C}$}{ 
  Solve $\mathrm{LP}(\bm{p})$ to obtain an optimal solution $(\bm{x}^*,\bm{y}^*)$\;
  Construct $\textsf{Candidates}(\bm{p})\coloneqq \{\textsf{Make\_it\_feasible}(S(r))\mid r\in \{y^*_v\mid v\in V\}\cup \{0\}\}$\;
  Take $\textsf{Best}(\bm{p})\in \argmax\{d(S)\mid S\in \textsf{Candidates}(\bm{p})\}$\;
  }
\Return $\argmax\{d(S)\mid S\in \{\textsf{Best}(\bm{p})\mid \bm{k}\leq \bm{p}\leq (|V_c|)_{c\in C}\}\}$\;
\end{algorithm}

\subsection{Analysis}
We first give a lower bound on the optimal value of $\mathrm{LP}(\bm{p})$. 
\begin{lemma}
\label{lem:LP_LB}
For any $S\subseteq V$ such that $|S_c|= p_c$ for every $c\in C$, there exists a feasible solution of $\, \mathrm{LP}(\bm{p})$ whose objective function value is greater than or equal to $d(S)$. 
\end{lemma}
\begin{proof}
Construct a solution $(\bm{x},\bm{y})$ of $\mathrm{LP}(\bm{p})$ as follows: 
\begin{align*}
x_e=
\begin{cases}
1/\|\bm{p}\|_1 &\text{if } e\in E(S),\\
0 &\text{otherwise}, 
\end{cases}
\qquad 
y_v=
\begin{cases}
1/\|\bm{p}\|_1 &\text{if } v\in S,\\
0 &\text{otherwise}. 
\end{cases}
\end{align*}
Then we can see that $(\bm{x},\bm{y})$ is feasible for $\mathrm{LP}(\bm{p})$. 
In fact, $\sum_{v\in V_c}y_v=\sum_{v\in S_c}y_v= \frac{p_c}{\|\bm{p}\|_1}$. 
The objective function value of $(\bm{x},\bm{y})$ is 
$\sum_{e\in E}x_e=\sum_{e\in E(S)}x_e=\frac{|E(S)|}{\|\bm{p}\|_1}=d(S).$
Thus, we have the lemma. 
\end{proof}

For $S\subseteq V$, let $C_\mathrm{sat}(S)=\{c\in C\mid |S_c|\geq k_c\}$, i.e., the set of colors for which the constraint is satisfied by $S$. 
We can prove the following key lemma (see Appendix~\ref{appendix:fair_third} for the proof): 

\begin{lemma}\label{lem:dalkks_key}
Let $S^*\subseteq V$ be an optimal solution to Problem~\ref{prob:dalkks}. 
For each $c\in C$, let $p^*_c = |S^*_c|$. 
Let $(\bm{x}^*,\bm{y}^*)$ be an optimal solution to $\mathrm{LP}(\bm{p}^*)$ and $\lambda$ its objective value.
Then there exists $S\in \{S(r)\mid r\in \{y^*_v\mid v\in V\}\cup \{0\}\}$ that satisfies (exactly) one of the following: 
\begin{enumerate}
    \item $d(S)\geq \lambda/3$ and $C_\mathrm{sat}(S)=C$; 
    \item $\displaystyle \frac{|E(S)|}{\sum_{c\in C_\mathrm{sat}(S)}|S_c|+\sum_{c\in C\setminus C_\mathrm{sat}(S)}k_c}\geq \lambda/3$ and $C_\mathrm{sat}(S)\neq \emptyset \neq C\setminus C_\mathrm{sat}(S)$; 
    \item $|E(S)|\geq |E(S^*)|/3$ and $C_\mathrm{sat}(S)=\emptyset$. 
\end{enumerate}
\end{lemma}

Based on the above lemma, we prove the following:  
\begin{theorem}
Algorithm~\ref{alg:dalkks} is a $1/3$-approximation algorithm for Problem~\ref{prob:dalkks} and runs in $O((n/|C|)^{|C|} T_\mathrm{LP})$ time, 
where $T_\mathrm{LP}$ is the time complexity required to solve the LP. 
\end{theorem} 
\begin{proof}
As the time complexity analysis is straightforward, we show the approximation ratio of $1/3$. 
Let $S^*\subseteq V$ be an optimal solution to Problem~\ref{prob:dalkks} and 
for each $c\in C$, let $p^*_c=|S^*_c|$. 
Looking at Lemma~\ref{lem:dalkks_key}, 
we see that when the algorithm solves $\mathrm{LP}(\bm{p}^*)$, 
one of the possible $S(r)$'s itself is just the subset of nodes $S\subseteq V$ whose existence is guaranteed in Lemma~\ref{lem:dalkks_key}. 
Then it is easy to see that 
$\textsf{Make\_it\_feasible}(S)$ is a $1/3$-approximate solution. 
Therefore, $\textsf{Best}(\bm{p}^*)$ is also a $1/3$-approximate solution. 
Noticing that the algorithm outputs the best of $\textsf{Best}(\bm{p})$'s for all possible $\bm{p}$ with $\bm{k}\leq \bm{p}\leq (|V_c|)_{c\in C}$ (containing $\bm{p^*}$), we are done. 
\end{proof}

Here we mention a simple acceleration technique for Algorithm~\ref{alg:dalkks}. 
As shown in Lemma~\ref{lem:LP_LB}, $\text{LP}(\bm{p}^*)$ with $p^*_c=|S^*_c|$ has the optimal value greater than or equal to $d(S^*)$. 
Therefore, if we solve $\text{LP}(\bm{p})$ for some $\bm{p}$ and its optimal value is less than the density of the current best feasible solution, the LP must not be $\text{LP}(\bm{p}^*)$ and thus we can skip the procedure to be applied to its optimal solution (i.e., Lines 3 and 4). 
In the next section, we present an additional acceleration technique that can reduce the number of LPs to solve.

\subsection{Acceleration via Greedy Peeling}
Our acceleration technique for Algorithm~\ref{alg:dalkks} is based on greedy peeling. 
Algorithm~\ref{alg:peeling} is a straightforward application of greedy peeling to Problem~\ref{prob:dalkks}, 
where for $S\subseteq V$ and $v\in S$, $\text{deg}_S(v)$ denotes the degree of $v$ in $G[S]=(S,E(S))$. 

\begin{algorithm}[t]
\caption{Greedy peeling algorithm for Problem~\ref{prob:dalkks}}\label{alg:peeling}
\SetKwInOut{Input}{Input}
\SetKwInOut{Output}{Output}
\Input{\ $G=(V,E)$, $\ell \colon V\rightarrow C$, $\bm{k}=(k_c)_{c\in C}\in \mathbb{Z}_{\geq 0}^{|C|}$}
\Output{\ $S\subseteq V$}
$S^{(n)}\leftarrow V$, $i\leftarrow n$\;
\While{$|S^{(i)}_c|>k_c\ $ for every $c\in C$ with $k_c\geq 1$}{ 
  $v_\text{min}\leftarrow \argmin_{v\in S^{(i)}}\degree_{S^{(i)}}(v)$\;
  $S^{(i-1)}\leftarrow S^{(i)}\setminus \{v_\text{min}\}$, $i\leftarrow i-1$\;
  }
\Return $\argmax\{d(S)\mid S\in \{S^{(n)},\dots, S^{(i)}\}\}$\;
\end{algorithm}

This algorithm achieves an approximation ratio of $1/2$ in a very specific case as follows: 
\begin{lemma}
Assume that there exists an optimal solution $S^*\subseteq V$ that satisfies $|S^*_c|>k_c$ for every $c\in C$ with $k_c\geq 1$. 
Then Algorithm~\ref{alg:peeling} outputs a $1/2$-approximate solution for Problem~\ref{prob:dalkks}. 
\end{lemma}
\begin{proof}
It is easy to see that for any $v\in S^*$, $S^*\setminus \{v\}$ is a feasible solution for Problem~\ref{prob:dalkks}. 
Therefore, we have that for any $v\in S^*$,
\begin{align*}
   d(S^*)\geq d(S^*\setminus \{v\}). 
\end{align*}
Transforming the above inequality, we have that for any $v\in S^*$, 
\begin{align}\label{ineq:optimality}
  \degree_{S^*}(v)\geq d(S^*).
\end{align}

Let $v^*$ be the node that is contained in $S^*$ and removed first by the algorithm. 
Note that the existence of such a node is guaranteed due to the assumption of the lemma and the design of the algorithm. 
Let $S'\subseteq V$ be the subset of nodes kept just before the removal of $v^*$. 
Obviously $S^*\subseteq S'$ and thus $S'$ is a feasible solution of Problem~\ref{prob:dalkks}. 
We can evaluate the density of $S'$ as follows:
\begin{align*}
    d(S') &= \frac{\frac{1}{2}\sum_{v\in S'}\degree_{S'}(v)}{|S'|}
    \geq \frac{1}{2}\degree_{S^*}(v^*)
    \geq \frac{1}{2}d(S^*),
\end{align*}
where the first inequality follows from the greedy choice of $v^*$ and the relation $S'\supseteq S^*$, 
and the second inequality follows from inequality \eqref{ineq:optimality}. 
Noticing that $S'$ is one of the output candidates of the algorithm, we have the lemma. 
\end{proof}

Finally, our accelerated algorithm is described in Algorithm~\ref{alg:dalkks_faster}. 

\begin{algorithm}[t]
\caption{Accelerated $1/3$-approximation algorithm for Problem~\ref{prob:dalkks}}\label{alg:dalkks_faster}
\SetKwInOut{Input}{Input}
\SetKwInOut{Output}{Output}
\Input{\ $G=(V,E)$, $\ell \colon V\rightarrow C$, $\bm{k}=(k_c)_{c\in C}\in \mathbb{Z}_{\geq 0}^{|C|}$}
\Output{\ $S\subseteq V$}
Run Algorithm~\ref{alg:peeling} and obtain its output $S_\text{peel}$\;
Run Algorithm~\ref{alg:dalkks} after replacing ``$\bm{k}\leq \bm{p}\leq (|V_c|)_{c\in C}$'' by ``$\bm{k}\leq \bm{p}\leq (|V_c|)_{c\in C}$ \textit{with} $p_c=k_c$ for some $c\in C$ with $k_c\geq 1$'' in Lines 1 and 5, and obtain its output $S_\text{LP}$\;
\Return $\argmax\{d(S)\mid S\in \{S_\text{peel}, S_\text{LP}\}\}$\;
\end{algorithm}

\begin{theorem}
Algorithm~\ref{alg:dalkks_faster} is a $1/3$-approximation algorithm for Problem~\ref{prob:dalkks} and runs in $O(|C|(n/|C|)^{|C|-1}T_\mathrm{LP})$ time, where $T_\mathrm{LP}$ is the time complexity required to solve the LP. 
\end{theorem}
\begin{proof}
The time complexity analysis is again trivial. 
Hence, in what follows, we guarantee the approximation ratio. 
If there exists an optimal solution $S^*\subseteq V$ that satisfies $|S^*_c|>k_c$ for every $c\in C$ with $k_c\geq 1$, then the output $S_\text{peel}$ of Algorithm~\ref{alg:peeling} is a $1/2$-approximate solution; thus, we are done. 
Otherwise there exists an optimal solution $S^*\subseteq V$ that satisfies $|S^*_c|=k_c$ for some $c\in C$ with $k_c\geq 1$. 
The modified version of Algorithm~\ref{alg:dalkks} used in Algorithm~\ref{alg:dalkks_faster} skips some $\bm{p}$'s with $\bm{k}\leq \bm{p}\leq (|V_c|)_{c\in C}$ but still tests $\bm{p}^*$ with $p^*_c=|S^*_c|$ for any $c\in C$. 
Therefore, the analysis similar to that for Algorithm~\ref{alg:dalkks} still works and we see that $S_\text{LP}$ is a $1/3$-approximate solution. 
Therefore, we have the theorem. 
\end{proof}

\section{Experimental Evaluation}\label{sec:experiments}

In this section, we evaluate our proposed algorithms using a variety of synthetic and real-world networks.

\subsection{Experimental Setup}
\label{subsec:setup} 

\begin{table}[t]
\caption{Dataset statistics summary.}
\label{tab:datasum}
\scalebox{0.86}{
\begin{tabular}{ccccc}
\toprule
Dataset & $|C|$ &$n$  & $m$         & $\alpha(V)$ \\
\midrule
Amazon Product &$2$    & 3.3K $\pm 4.7K$  & 24.7K $\pm 47.3K$ & 0.68 $\pm 0.12$   \\ 
Facebook100 &$2$ or $4$       & 10.7K $\pm 8.4K$ & 399K $\pm 315K$   & 0.48 $\pm 0.09$   \\ 
\midrule
GitHub Developers &$2$ & 37,700 &289,003   & 0.74                  \\  
LastFM Asia &$18$      & 7,624 & 27,806   & 0.20                 \\  
Deezer Europe &$2$     & 28,281 &92,752   & 0.56                  \\  
DBLP &$6$              & 25,176 &151,670   & 0.32                  \\  
\bottomrule
\end{tabular}
}
\end{table}

\noindent \textbf{Datasets.} 
We use two collections of datasets and four single-graph datasets, with attributes for the nodes in the graphs. 
Table~\ref{tab:datasum} summarizes the statistics of the following datasets, 
where for the first two datasets, we put the average values with the standard deviations.

 \begin{itemize}
 \leftskip=-10pt
 \item \textit{Amazon Product Metadata} - 299 graphs \cite{ni2019justifying}. This consists of a collection of 299 graphs, each with 2 colors, as curated by Anagnostopoulos et al.~\cite{anagnostopoulos2020spectral}. Classes are product categories, while edges indicate that products were co-purchased.

 \item \textit{Facebook100} - 100 graphs \cite{traud2012social}. This contains 100 anonymized networks of the Facebook social network from universities across the United States. Nodes have demographic attributes, like profession (faculty/student), gender (male/female), the year they joined the university, etc. For our purpose, we also create $2^2=4$ categories that combine profession and gender. The largest graph in terms of $m$ has 1,382,325 edges. 

 \item \textit{GitHub Developers} \cite{rozemberczki2019multiscale}. Nodes are developers in GitHub who have starred at least 10 repositories, and edges are mutual follower relationships between them. The attribute is whether the user is a web or a machine learning developer.

 \item \textit{LastFM Asia Social Network} \cite{feather}.  Nodes are LastFM users from Asian countries and edges are mutual follower relationships between them. The attribute is the location of a user.

 \item \textit{Deezer Europe Social Network} \cite{feather}. Nodes are Deezer users from European countries and edges are mutual follower relationships between them. The attribute is the gender of a user.

 \item \textit{DBLP Co-Authorship}. We also create a dataset from the DBLP database. We create a graph from authors of papers published in major conferences in $6$ areas (Theory, Data Management, Data Mining, Learning, Networking, and Image \& Video Processing) between 2003 and 2022. 
 Nodes are authors with attributes being the area (s)he has most publications. 
 Edges indicate that the two authors have collaborated at least once. 
 See Appendix~\ref{appendix:dblp} for details. 
\end{itemize}

\begin{figure*}[h]
    \centering
    \includegraphics[width=0.233\textwidth]{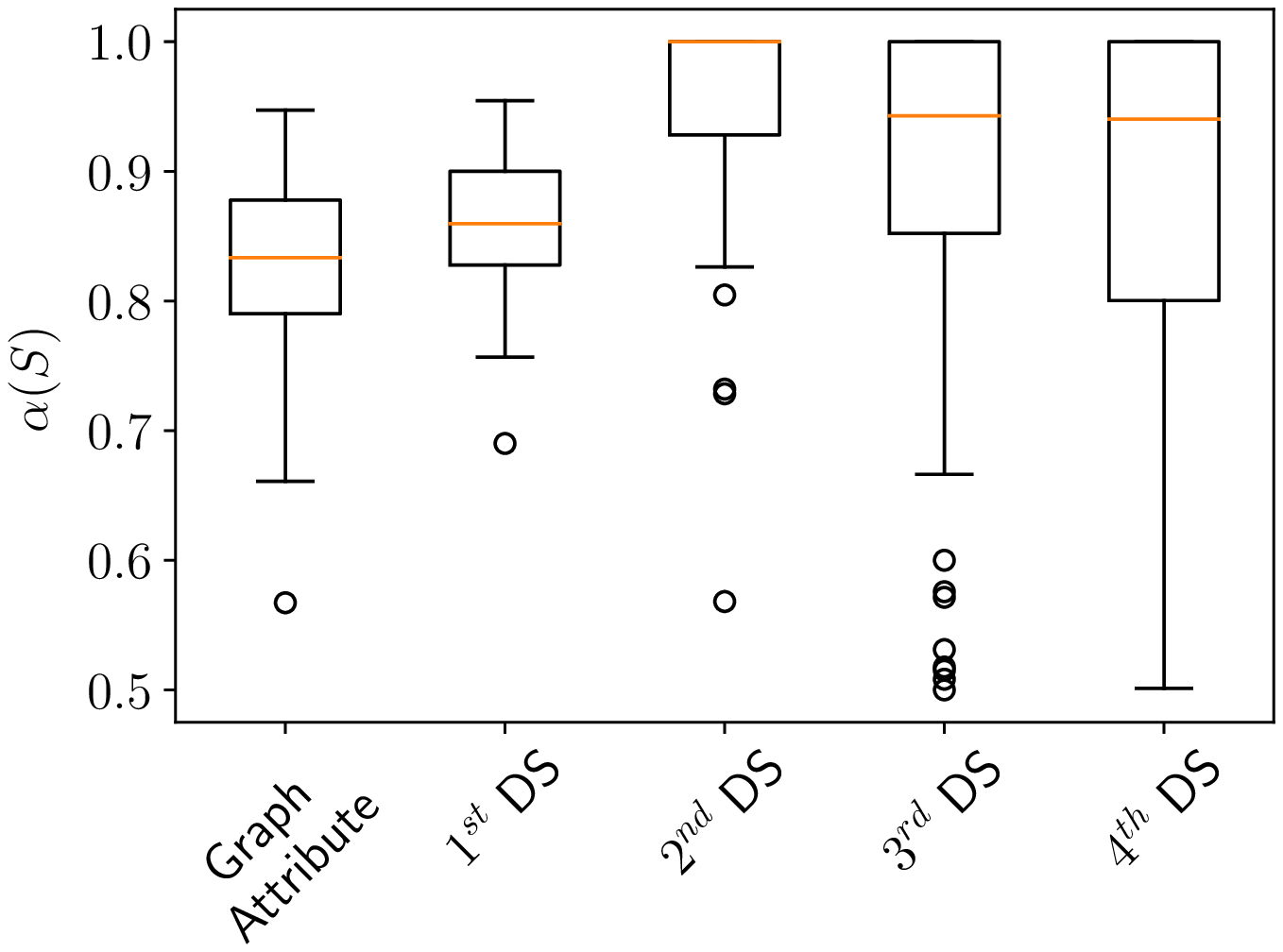}
    \includegraphics[width=0.233\textwidth]{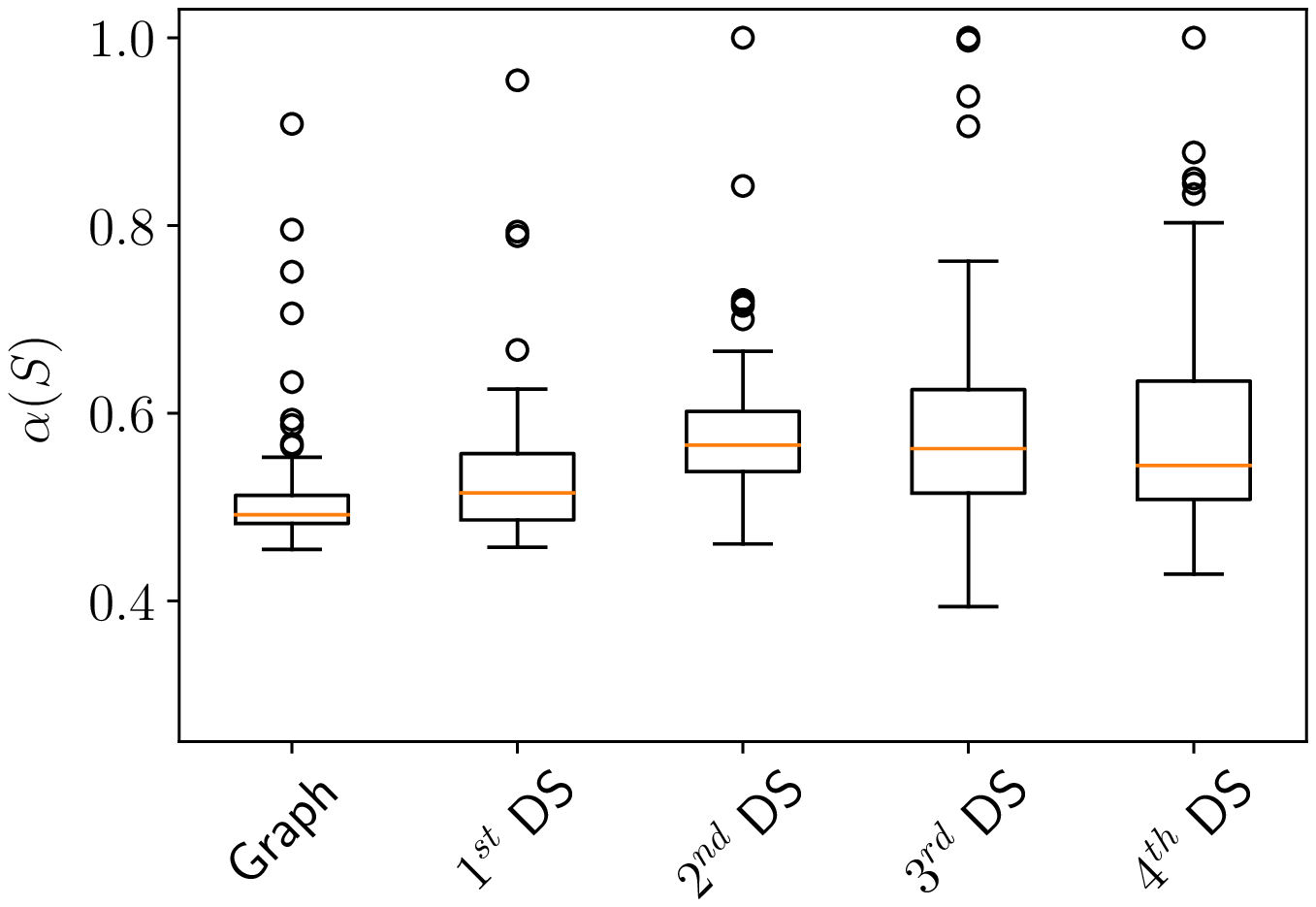}
    \includegraphics[width=0.233\textwidth]{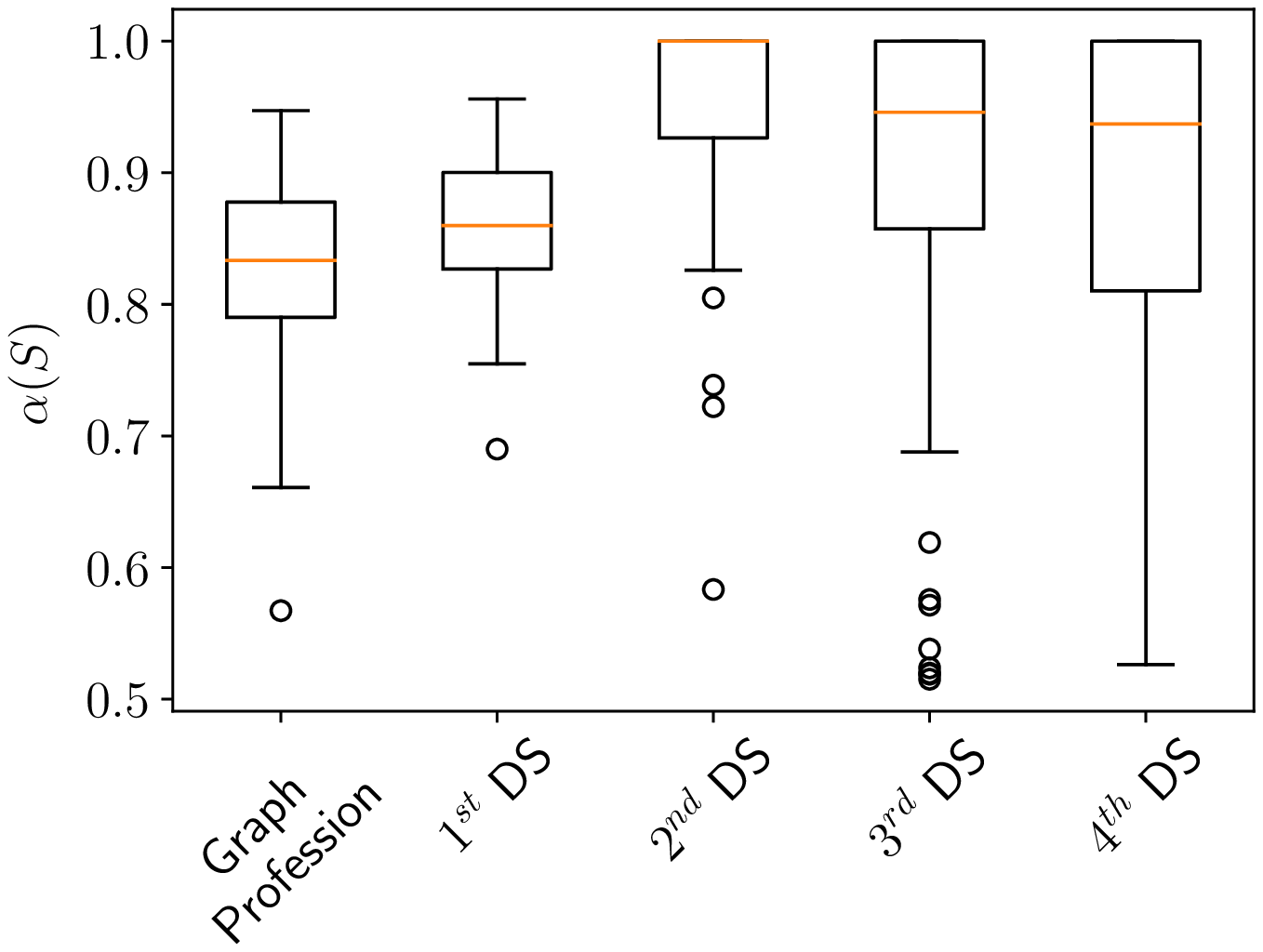} 
    \includegraphics[width=0.233\textwidth]{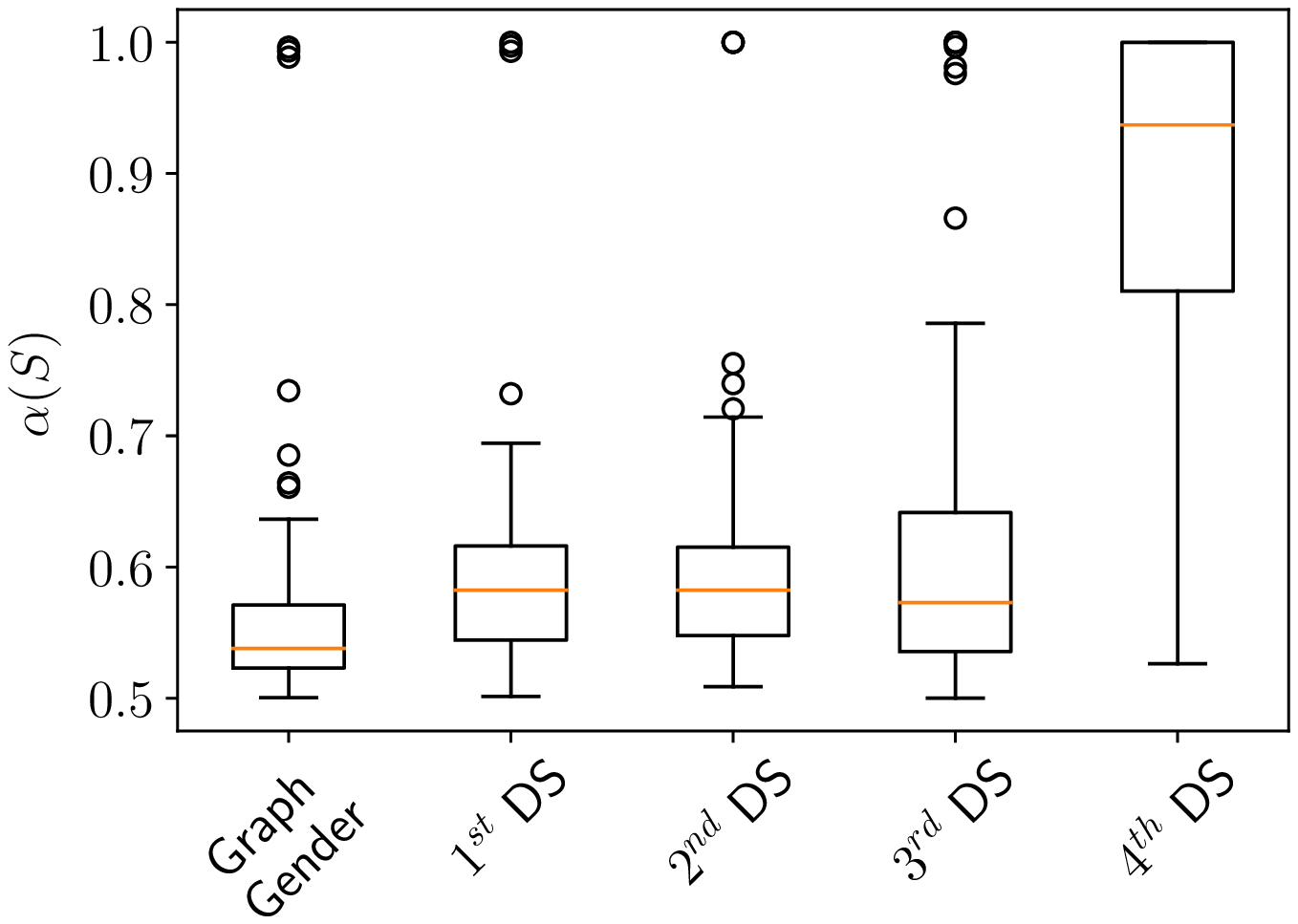}
    \caption{The diversity of the entire graph and the top-4 densest subgraphs of the first two datasets in Table~\ref{tab:datasum}. The first figure is for the Amazon Product dataset ($|C|=2$), while the latter three figures are for the Facebook100 dataset (left: Profession and Gender ($|C|=4$), middle: Profession ($|C|=2$), right: Gender ($|C|=2$)). Note that $\alpha(S)=1/|C|$ means the complete diversity of $S$.}
    \label{fig:dsd_statistics}
\end{figure*}

\smallskip
\noindent \textbf{Baselines.}
For Problem~\ref{prob:fair}, we employ the following baselines, including a novel baseline we devise using node embeddings~\cite{qiu2018network}: 
\begin{itemize}
\leftskip=-10pt
  \item \textit{DSP}. The (unconstrained) densest subgraph. Specifically, we use \textsc{Greedy++} by Boob et al. \cite{boob2020flowless} for 5 iterations, which in practice finds an optimal solution to DSP in very few iterations.
For $S\subseteq V$, we define the normalized density as its density divided by the optimal value to DSP.   

  \item \textit{PS} and \textit{FPS}. The spectral algorithms by Anagnostopoulos et al.~\cite{anagnostopoulos2020spectral} that split the entries of the largest eigenvector of the adjacency matrix (\textit{PS}) and ``fairified'' adjacency matrix (\textit{FPS}), based on their color, and sort them in a similar spirit with spectral clustering.

  \item  \textit{Embedding+Fair Clustering}. We design a novel baseline that first embeds the graph using the node embedding method NetMF \cite{qiu2018network} and then clusters the nodes using the fair $k$-means implementation of Backurs et al.~\cite{backurs2019scalable} for various values of $k$. We compute the density of each fair cluster for each value of $k$, and output the maximum among all.
\end{itemize}
For Problem~\ref{prob:dalkks}, we run the following two baselines: 
\begin{itemize}
   \leftskip=-10pt
\item \textit{IP.} We implement an exact algorithm using integer programming and test its scalability. See Appendix~\ref{appendix:IP} for details. 

\item \textit{Prop2.} We also implement the $1/4$-approximation algorithm introduced in Proposition~\ref{prop:baseline}. To solve Dal$k$S, we use the $1/2$-approximation algorithm based on LP with a greedy peeling acceleration~\cite{Khuller_Saha_09}. 
\end{itemize}

\smallskip
\noindent \textbf{Machine specs and code.} The experiments are performed on a single machine, with Intel i7-10850H CPU @ 2.70GHz and 32GB of main memory. We use C++ for all experiments. Linear programming and integer programming are implemented with SciPy and solved with HiGHS~\cite{huangfu2015highs}.
The code and datasets are available online.\footnote{https://github.com/tsourakakis-lab/densest-diverse-subgraphs}

\subsection{Evaluation of Algorithm~\ref{alg:approx} (for Problem~\ref{prob:fair})}


\smallskip
\noindent \textbf{Preliminary experiments.} Although we have created a reliable algorithmic toolbox for dense diverse clusters, we must assess their practical effectiveness. Specifically, when examining the densest subgraph or the top-$k$ densest subgraphs with sensitive attributes, do we observe diversity in the subgraphs? The answer is generally no, as observed in the vast majority of cases. This trend can result in unfair solutions, particularly in the context of sensitive attributes.  

In Figure~\ref{fig:dsd_statistics} we present the statistics of the entire graph and the top-4 densest subgraphs of the Amazon Product and Facebook100 datasets from Table~\ref{tab:datasum}, using standard box plots. 
As can be seen, the densest subgraphs tend to be not diverse. 
Especially for the Amazon Product dataset (leftmost figure), the densest subgraph most of the times is nearly monochromatic.  This phenomenon is not restricted to the densest subgraph but also found for the top-4 densest subgraphs, implying  that  strong homophily is a key factor behind dense clusters. We observe a similar trend for the Facebook100 dataset (the latter three figures). Even if we restrict our attention to the specific attribute we consider, profession or gender, we see that again they are not represented equally in the densest subgraph. These findings motivate the study of Problem~\ref{prob:fair}, where we can control the extent to which a specific color dominates the subgraph returned by an algorithm.

\smallskip
\noindent \textbf{Solution quality of algorithms.}
We compare Algorithm~\ref{alg:approx} with the baseline algorithms with respect to the density and diversity of the returned subgraphs, 
using the six datasets of Table~\ref{tab:datasum}.
We run Algorithm~\ref{alg:approx} by varying the value of $\alpha$. Note that when $\alpha(V)>\alpha$, Algorithm~\ref{alg:approx} may fail to find a feasible solution.  In that case, we resort to iteratively peeling a node with the minimum degree from $S$ with the most dominant color, until we obtain a feasible solution, or conclude that none exists. \textit{PS} and \textit{FPS} by Anagnostopoulos et. al. \cite{anagnostopoulos2020spectral} return by definition completely-balanced subgraphs, thus corresponding to the solutions of Problem~\ref{prob:fair} for a value of $\alpha = 1/|C|$. 

Our results are depicted in Figure~\ref{fig:dsd_examples} for the two collections of datasets, Amazon Product and Facebook100. 
As these two datasets consist of a number of graphs, we present the average values with the standard deviations. 
We see that Algorithm~\ref{alg:approx} outperforms the baselines. 
Let us first focus on the case of $\alpha= 1/|C|$. 
Algorithm~\ref{alg:approx} returns a better solution than that of \textit{PS} and \textit{FPS} for the Amazon Product dataset, while \textit{PS} and \textit{FPS} are comparable to Algorithm~\ref{alg:approx} for the Facebook100 dataset. We observe similar trends for the single graph datasets, as seen in Figure~\ref{fig:additional_datasets}. More precisely, when $\alpha = 1/|C|$, Algorithm~\ref{alg:approx} performs better than \textit{PS} and \textit{FPS} in $3$ out of $4$ cases.
Moreover, we see that by varying the parameter $\alpha$, we can adjust the trade-off between getting a subgraph as dense as possible (when we let $\alpha$ approach the diversity of the densest subgraph in the graph), or sacrifice density with the aim of getting more diverse subgraphs (completely balanced when we set $\alpha = 1/|C|$). 

\begin{figure}
    \centering
    \includegraphics[width=0.233\textwidth]{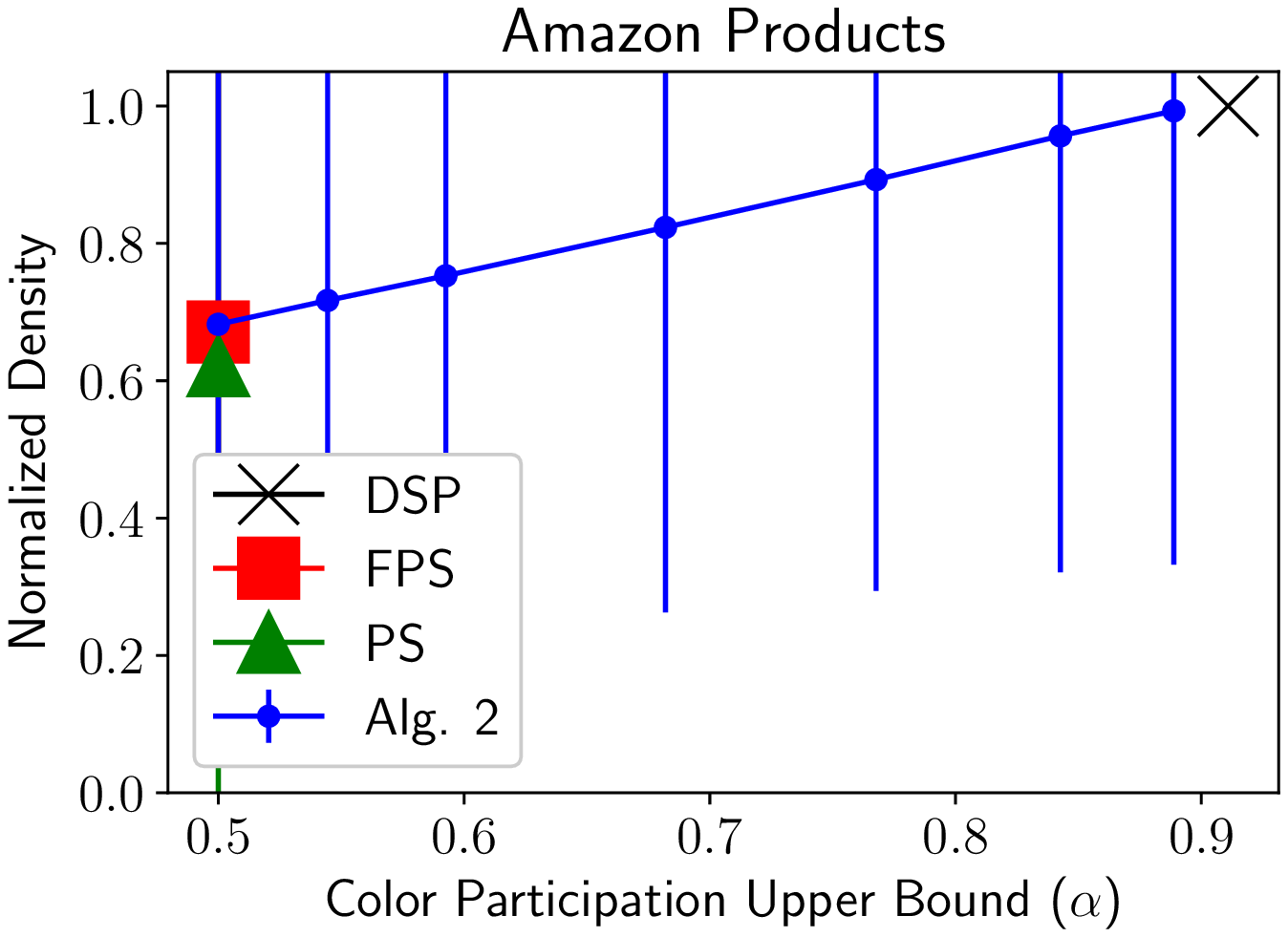}
    \includegraphics[width=0.233\textwidth]{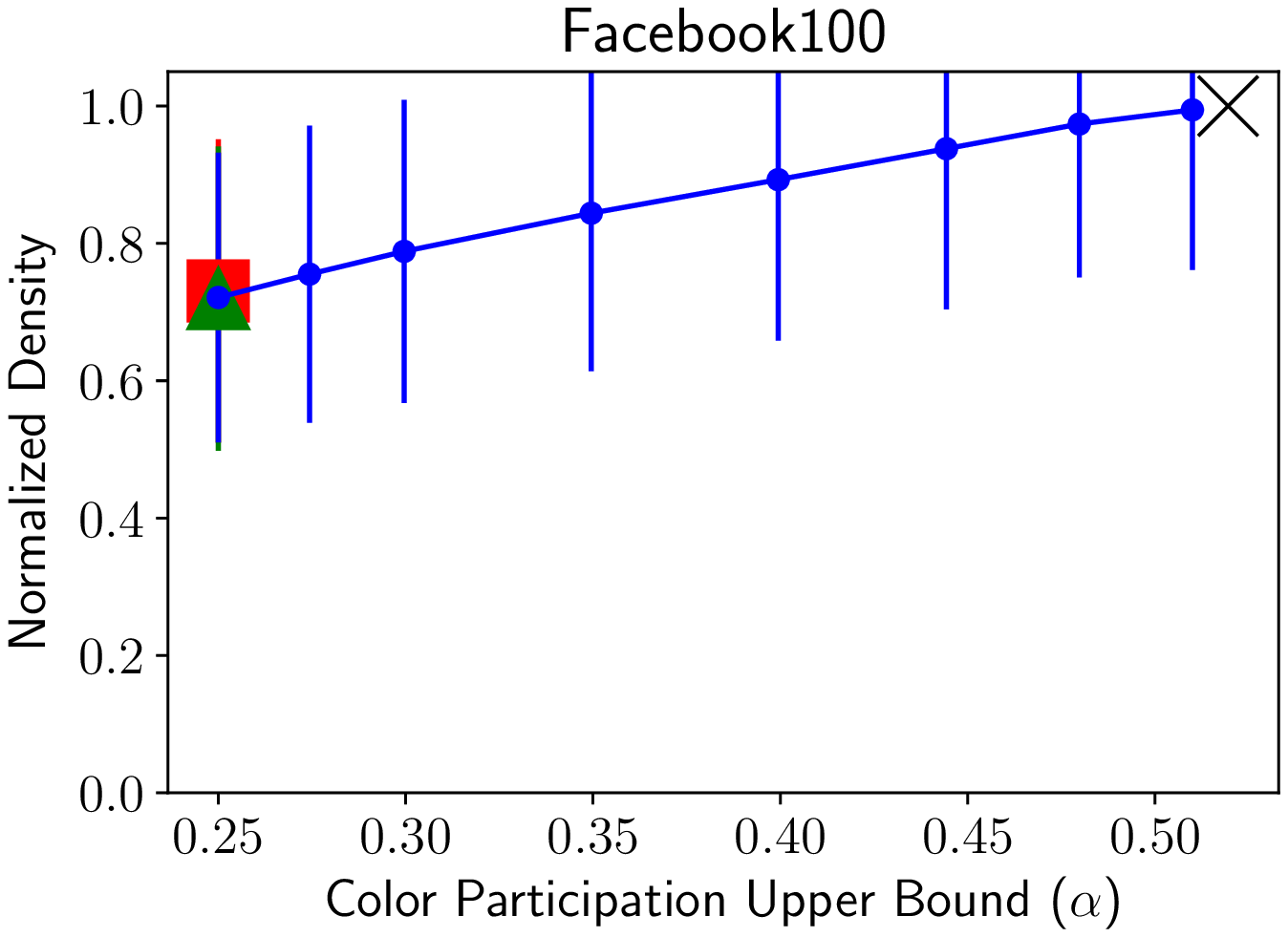}
    \caption{Performance of algorithms for the Amazon Products and Facebook100 datasets.}
    \label{fig:dsd_examples}
\end{figure}

\begin{figure*}[htb!]
        \centering
        
        \includegraphics[width=0.233\textwidth]{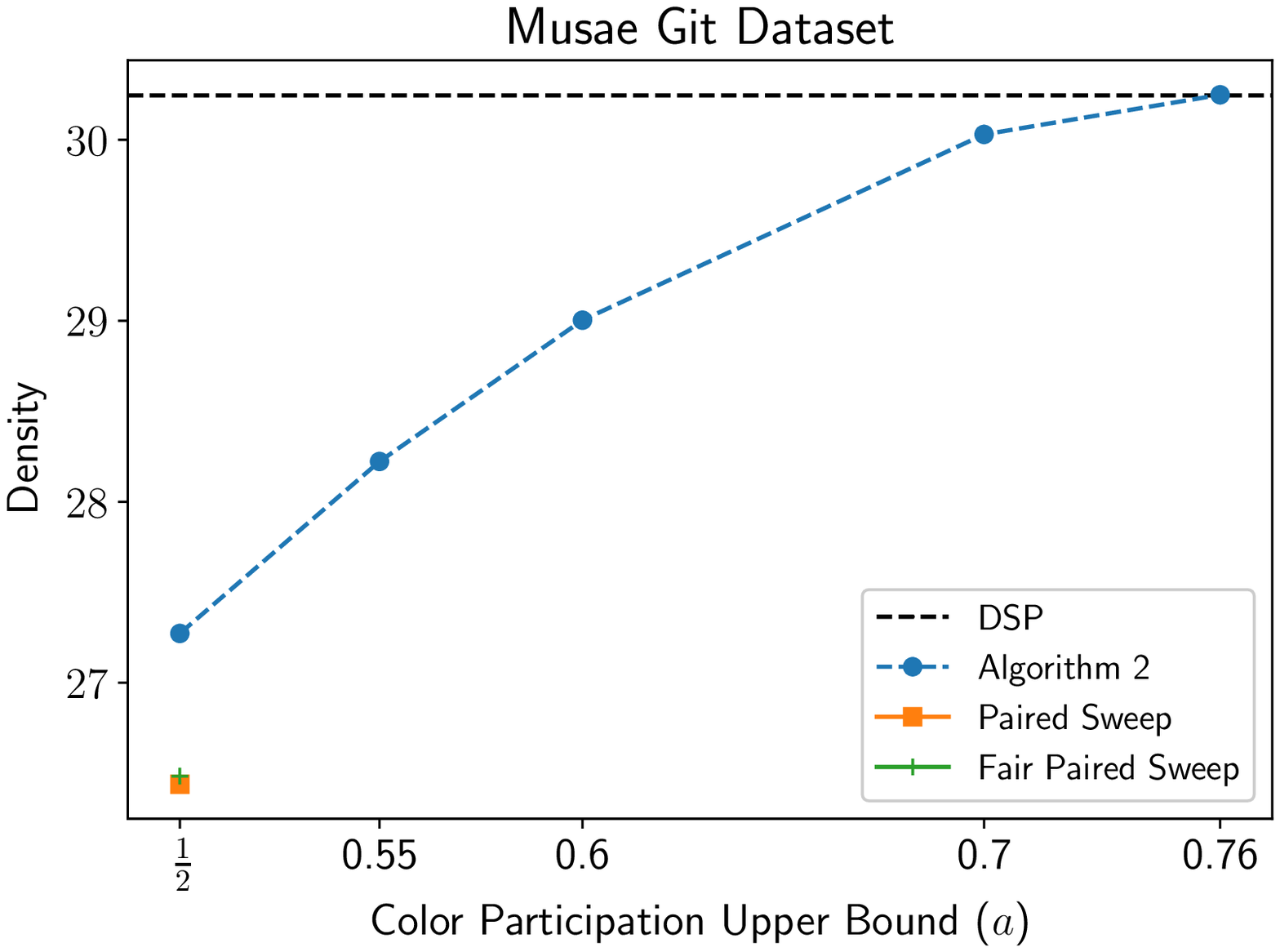}
            \includegraphics[width=0.233\textwidth]{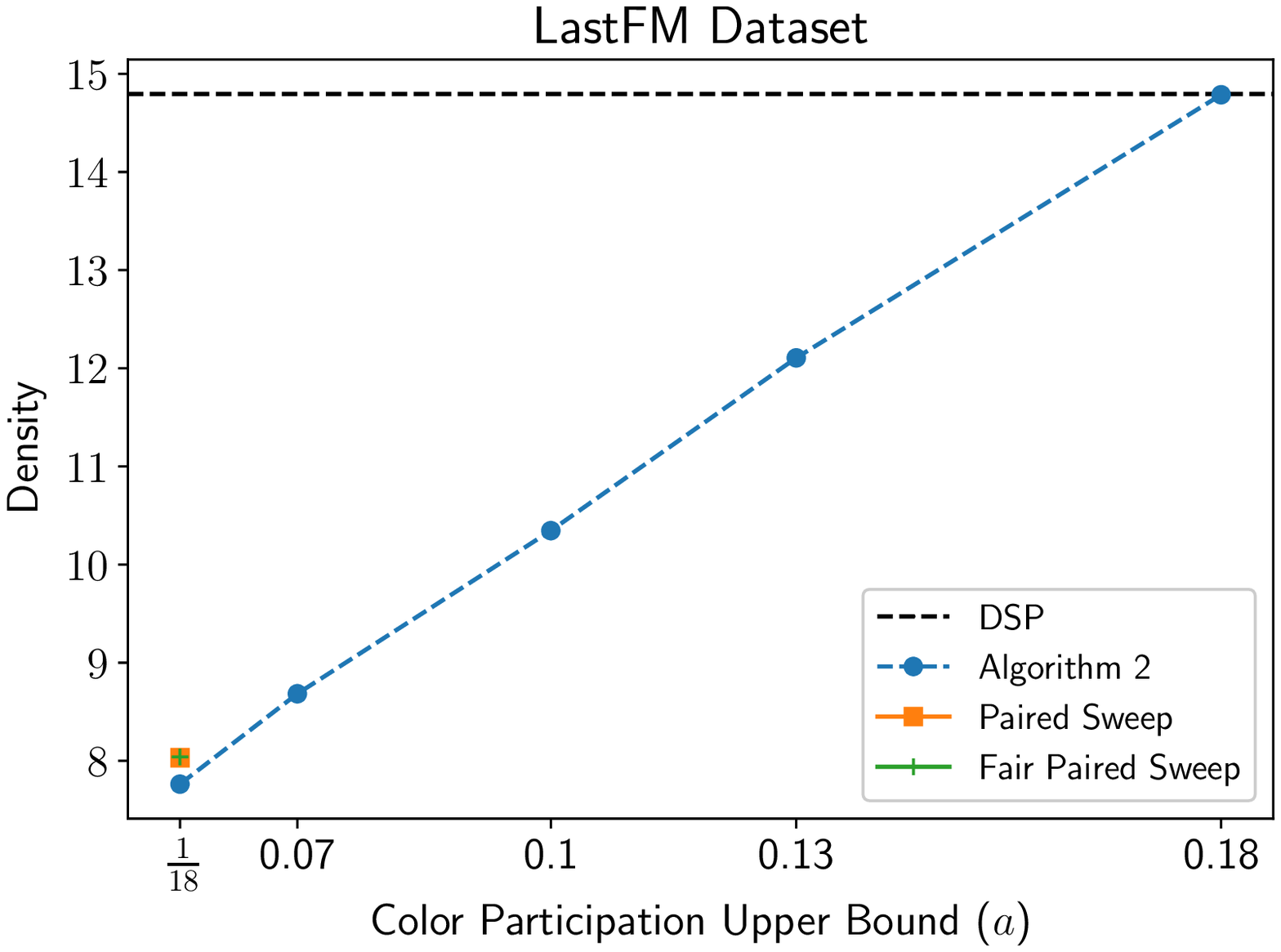}
             \includegraphics[width=0.233\textwidth]{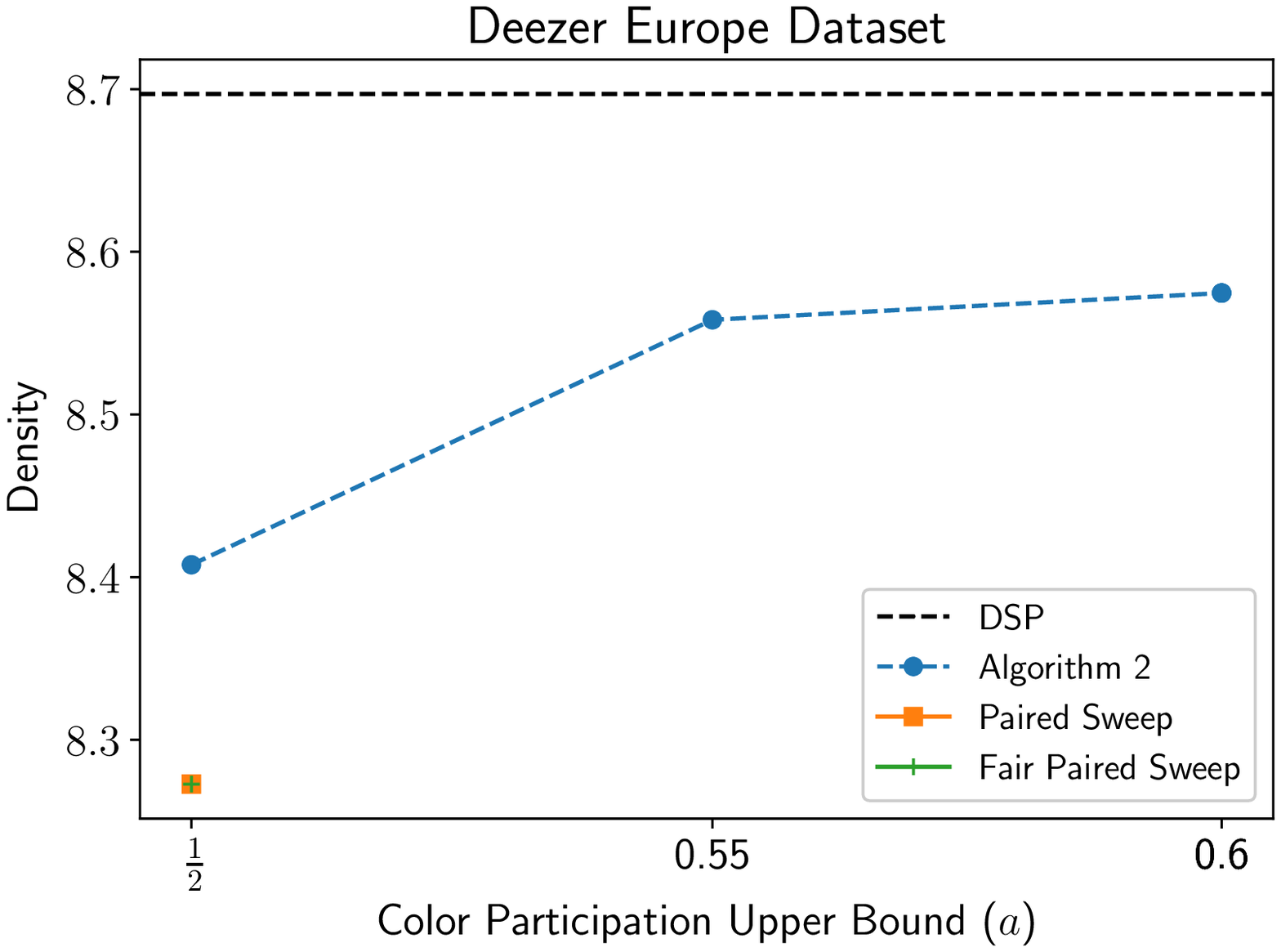}
            \includegraphics[width=0.233\textwidth] {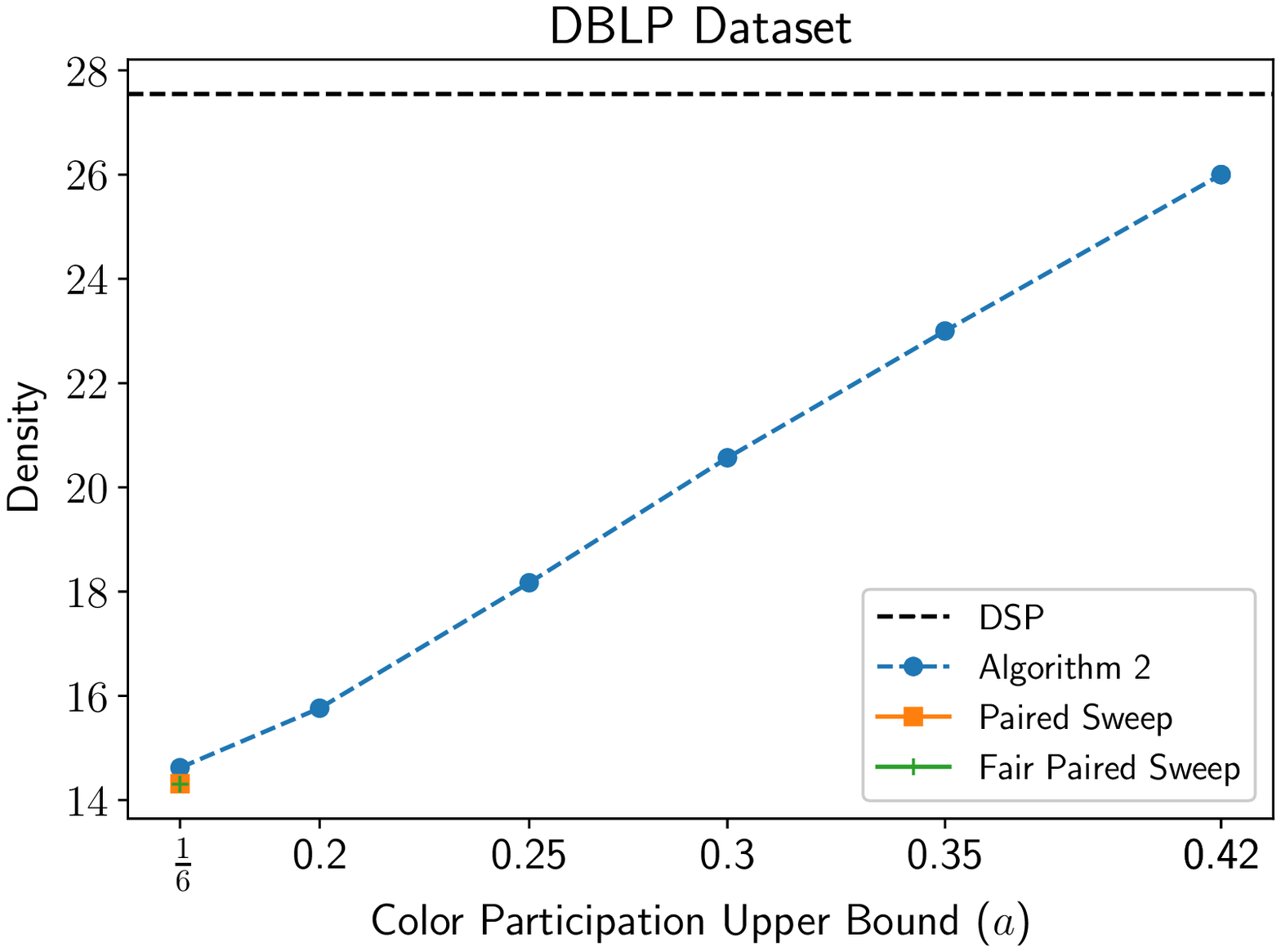}            
        \caption{Performance of algorithms for the single-graph datasets.}
       \label{fig:additional_datasets}
\end{figure*}

\smallskip
\noindent \textbf{(Poor) performance of \textit{Embedding+Fair Clustering}.}
Here we use a random sample of $35$ graphs from the Facebook100 dataset due to the cost of producing embeddings for each graph, as we did not notice significant changes after including more datasets in terms of results. We also use only gender as our attribute (as fair clustering methods allow only two colors), while we range $k$ in $k$-means from 2 to 32 and report the best result. In Figure ~\ref{fig:dsd_embeddings}, we see the results of this baseline in terms of the diversity of the clusters we get, as well their density. We see that simply performing $k$-means on the embedded graph results in clusters that are highly unbalanced. On the other hand, fair $k$-means allows us to get fair clusters by design (Figure ~\ref{fig:dsd_embeddings} (left)), but the corresponding subgraph is not as dense as the output by Algorithm~\ref{alg:approx} (Figure ~\ref{fig:dsd_embeddings} (right)). 
Notice that for the attribute we consider in this case, the ratio of the two colors in the graph is only slightly different from that of the densest subgraph (see the rightmost of Figure~\ref{fig:dsd_statistics}). Hence, Algorithm~\ref{alg:approx} finds subgraphs almost as dense as the densest.

\smallskip
\noindent \textbf{Running time of algorithms.}
In Figure~\ref{fig:runtimes} (left) we report the runtime of algorithms on the Amazon Product dataset. 
We employed this dataset because the sizes of the graphs (in terms of $n$ and $m$) vary substantially so that we can easily see how the running time of algorithms grows. We see that all of them scale (almost) linearly with the number of edges; in particular, the result for Algorithm~\ref{alg:approx} is consistent with the theoretical analysis in Section~\ref{subsec:algo_approx}. 
Note that for Algorithm~\ref{alg:approx}, we report the runtime for the most difficult case, when $\alpha=0.5$; hence the subgraphs need to be completely balanced. As we relax the diversity constraint, the runtime drops significantly as evident on the synthetic graphs, as shown in Figure~\ref{fig:runtimes} (right). The synthetic graph we employed resembles the stochastic block model paradigm. We constructed a graph with 5 clusters of 40,000 nodes each, with higher intra-connectivity within clusters, than inter-connectivity across clusters.  Most importantly, one of the clusters has a much higher density than the rest. Hence, the densest subgraph consists of ~40,000 nodes. Initially, we assign nodes to one of five colors uniformly at random. Hence, the (original) densest subgraph is already diverse. Therefore, Algorithm~\ref{alg:approx} does not need to invoke Procedure~\ref{alg:diversify}. Thus, relaxing the diversity constraint has no impact in this case (blue dotted line). On the other hand, if we assign a unique color to each cluster, then the densest subgraph becomes monochromatic and in order to have a diverse representation we need to involve the whole graph to our solution, resulting in a longer runtime that gets decreased as we relax the diversity constraint (solid orange line). 
Finally we remark that on the Facebook100 dataset 
Algorithm 2 runs in 3s for $\alpha = 0.5$ and 11s for $\alpha = 0.25$ even on the aforementioned largest graph with more than one million edges. 

\begin{figure}[t]
    \centering
    \includegraphics[width=0.233\textwidth]{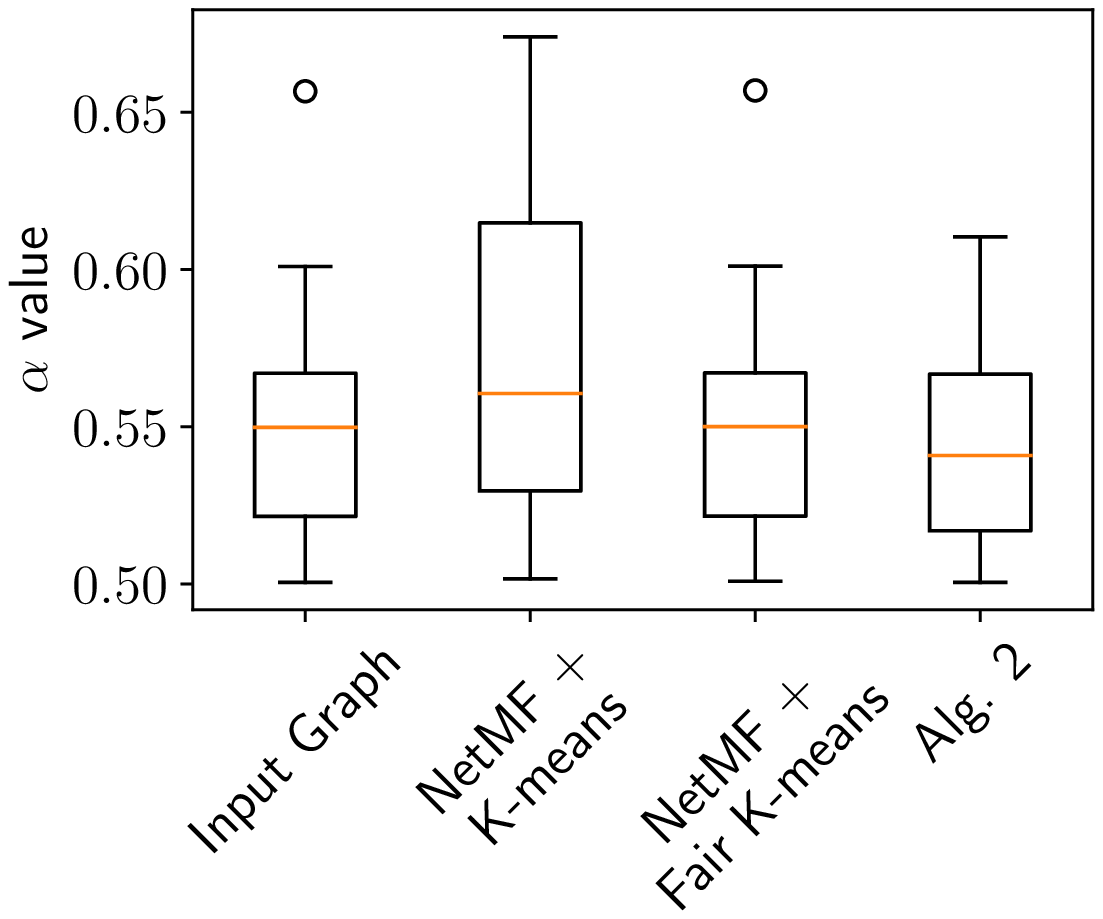}
    \includegraphics[width=0.233\textwidth]{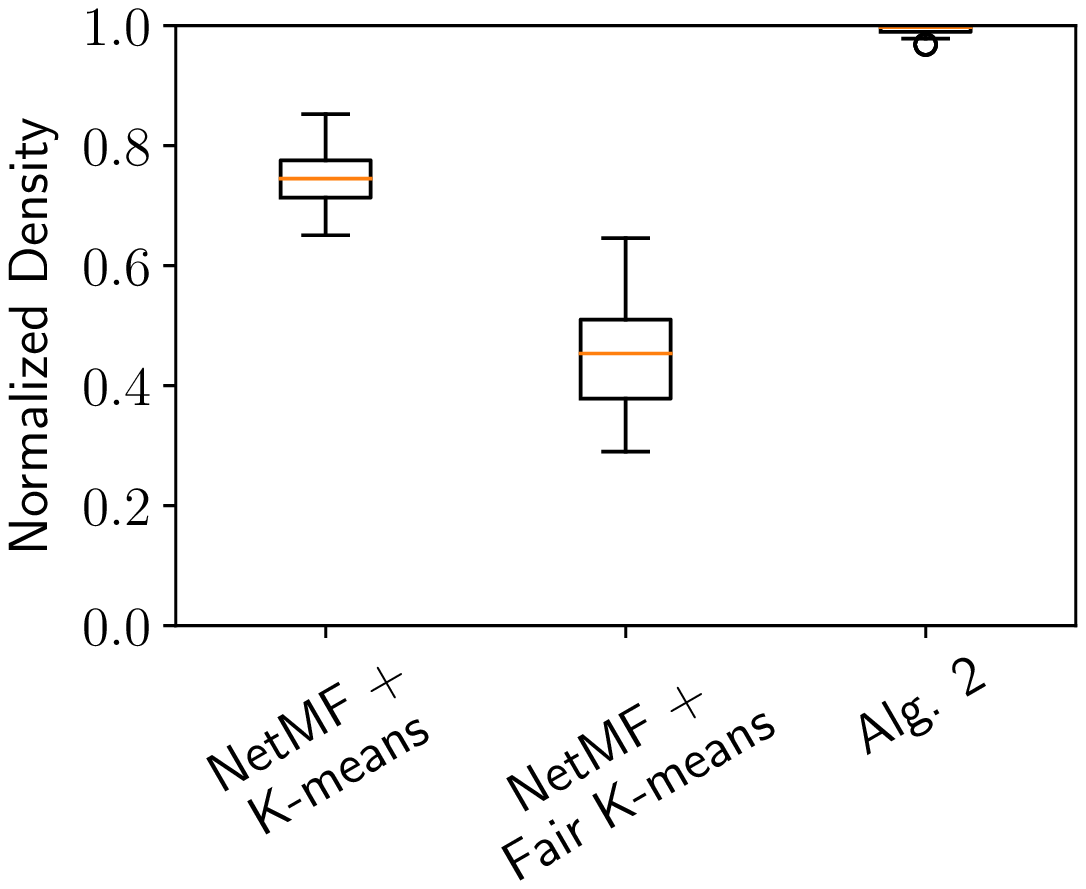}
    \caption{Diversity and normalized density of subgraphs obtained using \textit{Embedding+Fair Clustering} (together with its unfair variant), and Algorithm~\ref{alg:approx}.}
    \label{fig:dsd_embeddings}
\end{figure}

\begin{figure}[t]
    \centering
    \includegraphics[width=0.245\textwidth]{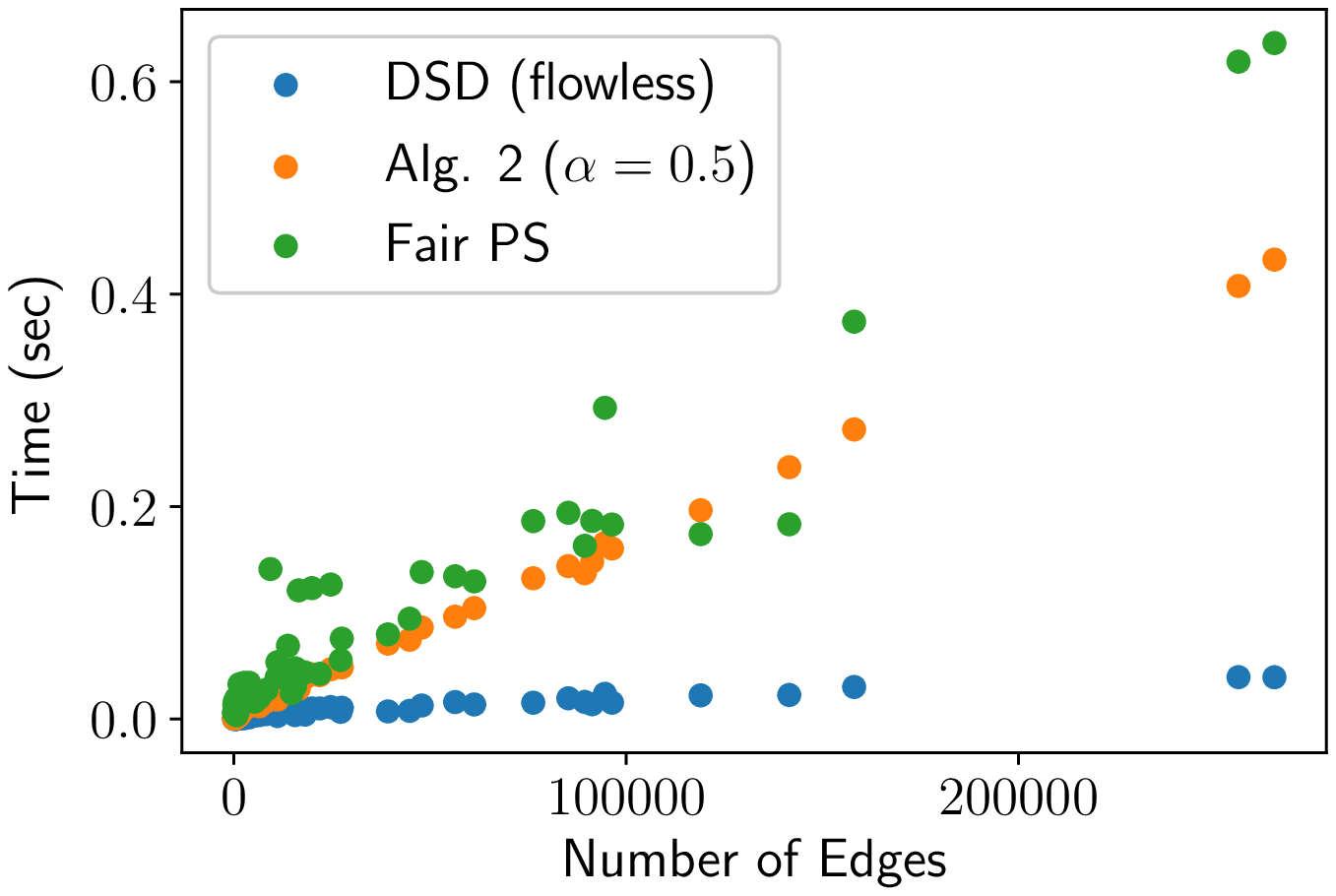}
    \includegraphics[width=0.22\textwidth]{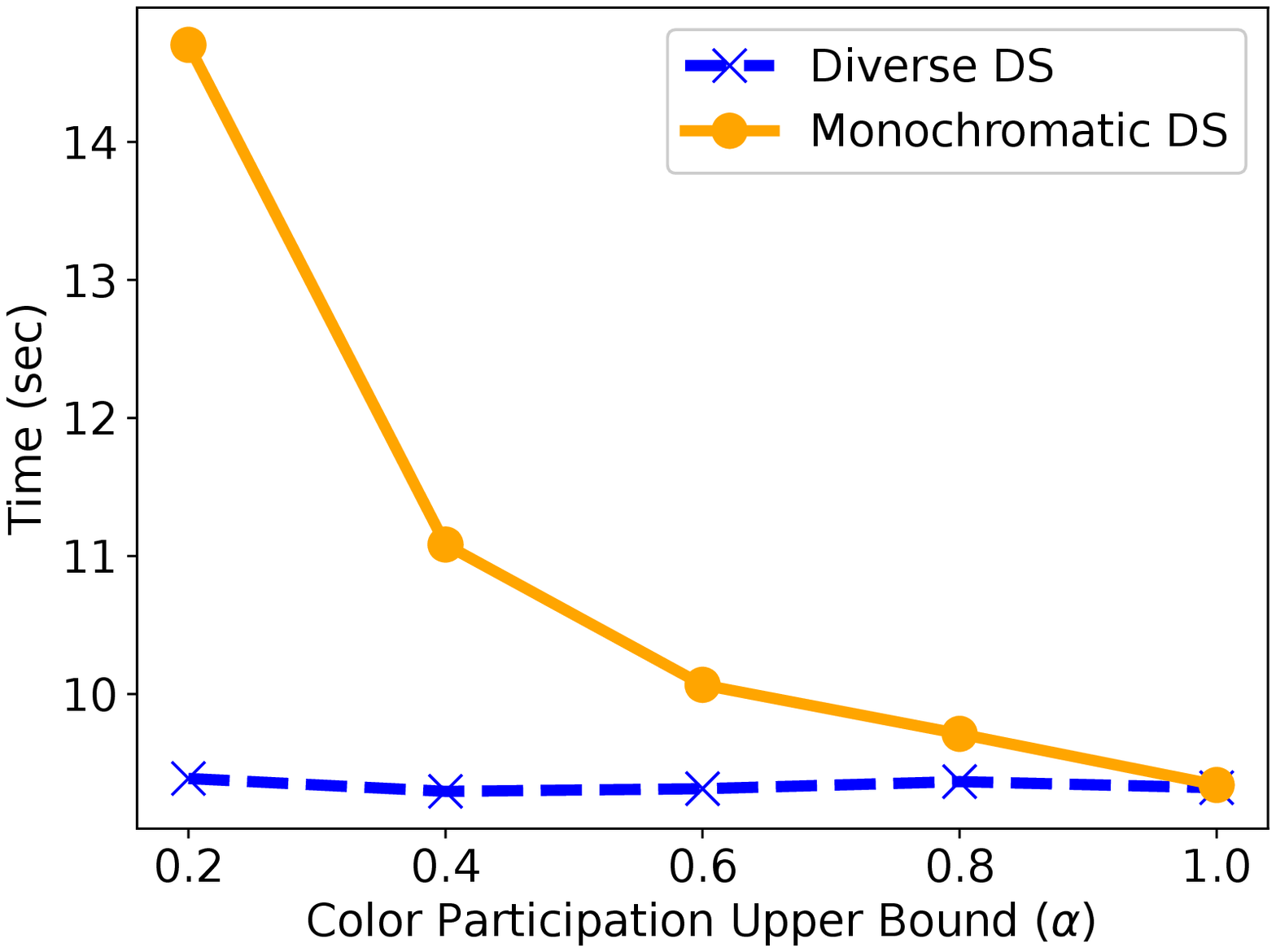}
    \caption{Left: Runtime as a function of the number of edges for all algorithms. Right: Runtime of Algorithm 2 as a function of the diversity parameter $\alpha$.}
    \label{fig:runtimes}
\end{figure}

\subsection{Evaluation of Algorithm~\ref{alg:dalkks_faster} (for Problem~\ref{prob:dalkks})}
We empirically study the accuracy and efficiency of Algorithm~\ref{alg:dalkks_faster} to understand its practical performance. We apply our algorithm and the baseline \textit{Prop2} on a part of the Amazon Product dataset, consisting of the graphs with at most 1,000 nodes and 5,000 edges. For each color $c\in C$, we set the lower bound as $k_c=|V_c|/2$. 
We also calculate the optimal solutions using \textit{IP} and report the empirical approximation ratios of Algorithm~\ref{alg:dalkks_faster} and \textit{Prop2}.

The results are shown in Figure~\ref{fig:eval_alg6}. In Figure~\ref{fig:eval_alg6} (left) we observe that in most graphs the solutions returned by Algorithm~\ref{alg:dalkks_faster} are optimal or near-optimal, with empirical approximation ratios consistently higher than $0.95$. This means that our proposed algorithm is of high accuracy in practice, despite the theoretical approximation guarantee being $1/3$. The baseline \textit{Prop2} returns solutions with empirical approximation ratios greater than $0.8$ but the performance is worse than that of ours. 
In Figure~\ref{fig:eval_alg6} (right) we show the running times of all algorithms applied. As expected, \textit{Prop2} is the fastest as it solves only one LP.  Algorithm~\ref{alg:dalkks_faster} can scale to graphs in two colors with thousands of nodes, and it is faster than \textit{IP} in general. 
Although from the figure the scalability of Algorithm~\ref{alg:dalkks_faster} looks not much better than that of \textit{IP}, 
we wish to remark that there are 7 graphs for which Algorithm 6 runs more than 10 times faster than \textit{IP}. 
In the extreme example, where $n = 30$ and $m = 381$, 
Algorithm~\ref{alg:dalkks_faster} runs in 0.58s, while \textit{IP} consumes 242.2s, 
meaning that Algorithm~\ref{alg:dalkks_faster} runs more than 400 times faster than \textit{IP}. 
We further discuss the scalability of Algorithm~\ref{alg:dalkks_faster} particularly with respect to $|C|$ in Appendix~\ref{appendix:synthetic}. 

\begin{figure}
    \centering
    \includegraphics[width=0.235\textwidth]{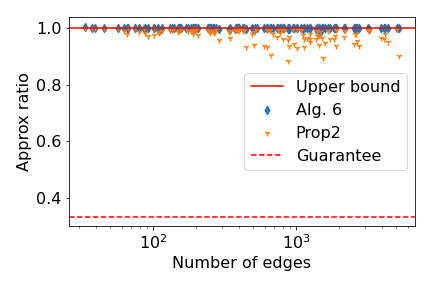}
    \includegraphics[width=0.235\textwidth]{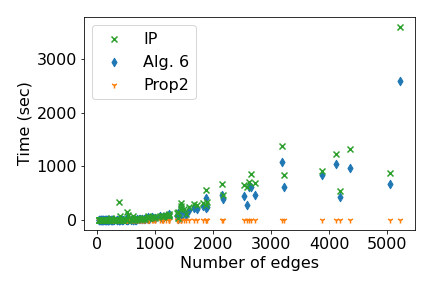} 
    \caption{Empirical approximation ratios and running times of Algorithm~\ref{alg:dalkks_faster} and baselines on the Amazon Product dataset with at most 1,000 nodes and 5,000 edges.} 
    \label{fig:eval_alg6}
\end{figure}

\section{Conclusions}

In this work, we have focused on the problem of finding a densest diverse subgraph. 
We proposed novel formulations and approximation algorithms for two different notions of diversity. 
We performed various experiments on synthetic and real-world datasets, verifying that the densest subgraphs tend to be driven by homophily and that our tools provide the state-of-the-art methods. 

Our work makes significant progress towards DSD with diversity and opens up several interesting problems. 
Can we improve the $\Omega(1/\sqrt{n})$-approximation for Problem~\ref{prob:fair} (in the case of $V$ being diverse)?  Can we design a better algorithm for Problem~\ref{prob:dalkks}, in terms of both the approximation ratio and the runtime? 
Investigating the hardness of approximation is also an interesting direction.

\bibliographystyle{ACM-Reference-Format}
\balance
\bibliography{ref}

\clearpage
\appendix

\section{Proof of Lemma~\ref{lem:dalkks_key}}\label{appendix:fair_third}
First we consider the case where $y_c\geq 1$ for every $c\in C$. 
Let $a$ be the maximum number that satisfies $C_\mathrm{sat}(S(a))=C$. 
Let $b$ be the infimum of the numbers $b$ that satisfy $C_\mathrm{sat}(S(b))=\emptyset$. 
Let $y^*_\mathrm{max}= \max_{v\in V}y^*_v$. 
Note that $0\leq a\leq b\leq y^*_\mathrm{max}$ hold. 

To prove the lemma, it suffices to show that (at least) one of the following cases occurs: 
\begin{description}
    \item[Case (i)] There exists some $r'\leq a$ such that $d(S(r'))\geq \frac{\lambda}{3}$;
    \item[Case (ii)] There exists some $a < r'\leq b$ such that \\ $\displaystyle \frac{|E(S(r'))|}{\sum_{c\in C_\mathrm{sat}(S(r'))}|S(r')_c|+\sum_{c\in C\setminus C_\mathrm{sat}(S(r'))}k_c}\geq \lambda/3$; 
    \item[Case (iii)] There exists some $r'> b$ such that $|E(S(r'))|\geq \frac{|E(S^*)|}{3}$. 
\end{description}

To show that, suppose that none of the above cases occurs. 
We define indicator functions $Z_v(r)\colon [0,y^*_\mathrm{max}]\rightarrow \{0,1\}$ for $v\in V$ and $Z_e(r)\colon [0,y^*_\mathrm{max}]\rightarrow \{0,1\}$ for $e=\{u,v\}\in E$ as follows: 
\begin{align*}
Z_v(r)=
\begin{cases}
1 &\text{if } r\leq y^*_v,\\
0 &\text{otherwise},
\end{cases}
\quad
Z_e(r)=
\begin{cases}
1 &\text{if } r\leq \min\{y^*_u,y^*_v\},\\
0 &\text{otherwise}.
\end{cases}
\end{align*}
Since Case (i) does not hold, for any $r'\leq a$, we have $d(S(r'))=\frac{|E(S(r'))|}{|S(r')|}< \frac{\lambda}{3}$. Thus, we have 
\begin{align*}
\int_0^a |E(S(r))|\rd r
&<\frac{\lambda}{3}\int_0^a |S(r)|\rd r \\
&=\frac{\lambda}{3}\int_0^a \sum_{v\in V}Z_v(r)\rd r
=\frac{\lambda}{3}\sum_{v\in V} \int_0^a Z_v(r)\rd r. 
\end{align*}
Since Case (ii) does not occur, for any $a < r' \leq b$, we have 
$|E(S(r'))|< \frac{\lambda}{3} \left(\sum_{c\in C_\mathrm{sat}(S(r'))}|S(r')_c|+\sum_{c\in C\setminus C_\mathrm{sat}(S(r'))}k_c\right)$. Hence, we have 
\begin{align*}
\int_a^b|E(S(r))|\rd r 
<\frac{\lambda}{3}\int_a^b\left(\sum_{c\in C_\mathrm{sat}(S(r))}|S(r)_c|+\sum_{c\in C\setminus C_\mathrm{sat}(S(r))}k_c\right)\rd r. 
\end{align*}
Now we see that 
\begin{align*}
&\int_a^b \sum_{c\in C_\mathrm{sat}(S(r))}|S(r)_c|\rd r
\leq \int_a^b \sum_{c\in C}|S(r)_c|\rd r \\
&= \int_a^b \sum_{c\in C}\sum_{v\in V_c}Z_v(r)\rd r 
= \sum_{c\in C}\sum_{v\in V_c} \int_a^b Z_v(r)\rd r\\
&= \sum_{v\in V} \int_a^b Z_v(r)\rd r
\end{align*}
and
\begin{align*}
\int_a^b \sum_{c\in C\setminus C_\mathrm{sat}(S(r))}k_c \rd r
\leq \int_a^b \sum_{c\in C}p^*_c \rd r
\leq \|\bm{p}^*\|_1\cdot y^*_\mathrm{max}\leq 1. 
\end{align*}
Therefore, we have 
\begin{align*}
\int_a^b|E(S(r))|\rd r < \frac{\lambda}{3}\left(\sum_{v\in V} \int_a^b Z_v(r)\rd r + 1\right). 
\end{align*}
Since Case (iii) does not hold, for any $r'> b$, we have $|E(S(r'))|< \frac{|E(S^*)|}{3}$. Thus, we have 
\begin{align*}
\int_b^{y^*_\mathrm{max}} |E(S(r))|\rd r
&< \frac{|E(S^*)|}{3} \int_b^{y^*_\mathrm{max}}\rd r\\
&\leq \frac{|E(S^*)|}{3}y^*_\mathrm{max} 
\leq \frac{1}{3}\frac{|E(S^*)|}{\|\bm{p}^*\|_1}\leq \frac{\lambda}{3}, 
\end{align*}
where the last inequality follows from Lemma~\ref{lem:LP_LB}. 
Thus, we have 
\begin{align*}
&\int_0^{y^*_\mathrm{max}}|E(S(r))|\rd r \\
&=\int_0^{a}|E(S(r))|\rd r 
+\int_a^{b}|E(S(r))|\rd r 
+\int_b^{y^*_\mathrm{max}}|E(S(r))|\rd r\\
&< \frac{\lambda}{3}\left(\sum_{v\in V} \int_0^a Z_v(r)\rd r 
+ \left(\sum_{v\in V} \int_a^b Z_v(r)\rd r + 1\right)
+ 1 \right) \\
&\leq \frac{\lambda}{3}\left(\sum_{v\in V}y^*_v + 2\right)
= \frac{\lambda}{3}\left(\sum_{c\in C}\sum_{v\in V_c}y^*_v + 2\right) = \lambda.
\end{align*}

On the other hand, by a simple calculation, we have 
\begin{align*}
&\int_0^{y^*_\mathrm{max}}|E(S(r))|\rd r
=\int_0^{y^*_\mathrm{max}}\sum_{e\in E}Z_e(r)\rd r\\
&=\sum_{e\in E} \int_0^{y^*_\mathrm{max}} Z_e(r)\rd r
=\sum_{e\in E} \min\{y^*_u,y^*_v\}
\geq \sum_{e\in E} x^*_e
=\lambda, 
\end{align*}
which contradicts to the above inequality. 

Next we consider the case where $y_c=0$ for some $c\in C$. 
Defining $a$, $b$, and $y^*_\mathrm{max}$ in the same way, we have $0\leq a \leq y^*_\mathrm{max}<b=\infty$. 
It suffices to show that (at least) one of the following cases occurs: 
\begin{description}
    \item[Case (i)] There exists some $r'\leq a$ such that $d(S(r'))\geq \frac{\lambda}{3}$;
    \item[Case (ii)] There exists some $a < r'\leq y^*_\mathrm{max}$ such that \\ $\displaystyle \frac{|E(S(r'))|}{\sum_{c\in C_\mathrm{sat}(S(r'))}|S(r')_c|+\sum_{c\in C\setminus C_\mathrm{sat}(S(r'))}k_c}\geq \lambda/3$. 
\end{description}
The proof is similar to the above and thus omitted. 
\qed

\section{Supplementary Material for Experiments}

\subsection{DBLP Co-Authorship Dataset}\label{appendix:dblp}
We create this dataset from all authors that have published at least 3 papers in the following conferences between 2003 and 2022: 
\begin{itemize}
\leftskip=-10pt
    \item \textbf{Theory: } COLT, FOCS, ICALP, SODA, STOC     
    \item \textbf{Data Management: } CIDR, ICDE, PODS, SIGMOD, VLDB
    \item \textbf{Data Mining: } CIKM, ICDM, KDD, WSDM, WWW   
    \item \textbf{Learning: } AAAI, ICLR, ICML, NeurIPS
    \item \textbf{Networking: } ICC, IMC, INFOCOM, MOBICOM, SIGCOMM
    \item \textbf{Image \& Video Processing: } CVPR, ECCV, ICCV
\end{itemize}

\subsection{\textit{IP} for Problem~\ref{prob:dalkks}}\label{appendix:IP}
Let us introduce a $0$--$1$ variable $x_e$ for each $e\in E$ and $0$--$1$ variable $y_v$ for each $v\in V$. Then for each $k_\text{guess}\in \{\|\bm{k}\|_1,\dots,  n\}$, we construct the following (mixed) integer linear programming problem: 
\begin{alignat*}{3}
\text{maximize}&&\quad & \sum_{e\in E} x_e/k_\text{guess}&\\
\text{subject to}&&      &\sum_{v \in V} y_v  = k_\text{guess},    \\
           &      &      &   \sum_{v \in V_c} y_v  \geq k_c  &\quad & \forall c\in C, \\
           &      &      &x_e\leq y_u,\ x_e\leq y_v && \forall e=\{u,v\}\in E, \\
           &      &      &0\leq x_e \leq 1 && \forall e\in E, \\
           &      &      &y_v\in \{0,1\} && \forall v\in V.
\end{alignat*}
It is easy to see that this problem is equivalent to Problem~\ref{prob:dalkks} with size constraint $|S|=k_\text{guess}$. Note that the $0$--$1$ constraints for $x_e$'s are relaxed because they are redundant. 

The baseline \textit{IP} computes an optimal solution to Problem~\ref{prob:dalkks} as follows: 
Solve the above integer programming problem for all $k_\text{guess}\in \{\|\bm{k}\|_1,\dots,  n\}$ and return a solution with the maximum density value.

\subsection{Scalability of Algorithm~\ref{alg:dalkks_faster} on Synthetic Data}\label{appendix:synthetic}
We run Algorithm~\ref{alg:dalkks_faster} on a set of Erd{\H{o}}s-R{\'e}nyi graphs with size $n$ ranged in $\{18,54,90,126\}$, edge probability $p=5/n$, and the number of colors $|C|$ ranged in $\{2,3,6\}$. Nodes are evenly split into colors, and the lower bound is set as $k_c=\left \lfloor\frac{n}{2|C|}\right \rfloor$ for each $c\in C$. Random graphs are sampled ten times for each setting, and the averaged running times are reported in Figure~\ref{fig:er_alg6}. When we have only two colors, Algorithm~\ref{alg:dalkks_faster} can scale up to a thousand nodes as shown in the main text. However, the algorithm becomes rather inefficient as the graph size goes beyond that or the number of colors increases. 
\begin{figure}[ht]
    \centering
    \includegraphics[width=.3\textwidth]{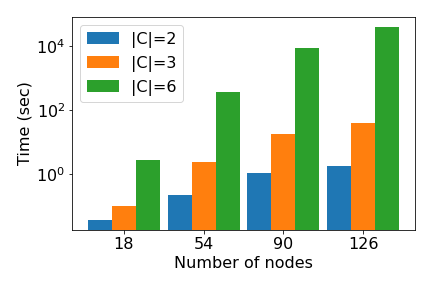}
    \caption{Running time of Algorithm~\ref{alg:dalkks_faster}.}
    \label{fig:er_alg6}
\end{figure}

\end{document}